\RequirePackage[l2tabu, orthodox]{nag} % Warn about obsolete LaTeX usage

\newif\ifArxiv
\Arxivtrue
\newif\ifExtended
\Extendedtrue
\ifArxiv
\Extendedtrue
\fi

\documentclass{AIMS}

\usepackage[utf8]{inputenc}
\usepackage[english]{babel}  % better hyphenation
\usepackage[T2A,T1]{fontenc} %T2A for Cyrillic symbols
\usepackage{amsmath,amssymb} % math symbols
\usepackage{graphicx}        % the LaTeX graphic package
\usepackage{color}           % colourise letters
\usepackage{varioref}        % for \vref{:...} gives "As 1.2 on page 4"
\usepackage{xspace}          % for \xspace = possible space after commands
\usepackage{cite}            % sort and compress simultaneous citations
\usepackage{multirow}        % for multirow and multicolumn commands
\usepackage{booktabs} %beautiful tables: toprule, midrule, bottomrule
\usepackage{stmaryrd} %for power series brackets
\usepackage{algorithm}
\usepackage{algpseudocode}
\usepackage{microtype} %More beautiful kerning
\usepackage{hyperref}
\usepackage{nicefrac}
\hypersetup{urlcolor=blue, citecolor=red}

\usepackage{enumitem} % to get local control over enums
\usepackage{mathtools} %for the power-on-the-left used by \rev

\ifExtended
\definecolor{navy}{rgb}{0,0,.5}
\else
\definecolor{navy}{rgb}{0,0,0}
\fi

\usepackage[caption=false]{subfig} % support for sub-figures
\usepackage{bm}

\usepackage{algorithm}
\usepackage{algpseudocode}

\usepackage{IEEEtrantools}

\def\currentvolume{X}
 \def\currentissue{X}
  \def\currentyear{20xx}
   \def\currentmonth{XX}
    \def\ppages{X--XX}
     \def\DOI{10.3934/amc.xx.xx.xx}

\newtheorem{theorem}{Theorem}[section]
\newtheorem{corollary}[theorem]{Corollary}
\newtheorem{lemma}[theorem]{Lemma}
\newtheorem{proposition}[theorem]{Proposition}

\newtheorem{problem}[theorem]{Problem}
\theoremstyle{definition}
\newtheorem{definition}[theorem]{Definition}
\newtheorem{remark}[theorem]{Remark}

\usepackage{fancyref}

\makeatletter
\def\mkfancyprefix#1#2{%
\expandafter\def\csname fancyref#1labelprefix\endcsname{#1}%
\begingroup\def\x{\endgroup\frefformat{plain}}%
    \expandafter\x\csname fancyref#1labelprefix\endcsname
    {\MakeLowercase{#2}\fancyrefdefaultspacing##1}%
\begingroup\def\x{\endgroup\Frefformat{plain}}%
    \expandafter\x\csname fancyref#1labelprefix\endcsname
    {#2\fancyrefdefaultspacing##1}%
\begingroup\def\x{\endgroup\frefformat{vario}}%
    \expandafter\x\csname fancyref#1labelprefix\endcsname
    {\MakeLowercase{#2}\fancyrefdefaultspacing##1##3}%
\begingroup\def\x{\endgroup\Frefformat{vario}}%
    \expandafter\x\csname fancyref#1labelprefix\endcsname
    {#2\fancyrefdefaultspacing##1##3}%
}
\makeatother
\fancyrefchangeprefix{\fancyrefeqlabelprefix}{eqn}
\Frefformat{vario}{\fancyrefeqlabelprefix}{(#1)#3}%fix vref of eqn.

\mkfancyprefix{ssec}{Section}
\mkfancyprefix{tbl}{Table}
\mkfancyprefix{thm}{Theorem}
\mkfancyprefix{lem}{Lemma}
\mkfancyprefix{cor}{Corollary}
\mkfancyprefix{prop}{Proposition}
\mkfancyprefix{prob}{Problem}
\mkfancyprefix{alg}{Algorithm}
\mkfancyprefix{inv}{Invariant}
\mkfancyprefix{ex}{Example}
\mkfancyprefix{line}{Line}
\mkfancyprefix{def}{Definition}
\mkfancyprefix{itm}{Item}
\mkfancyprefix{app}{Appendix}
\newcommand{\cref}[1]{\Fref{#1}}

\ifArxiv
\else
\AtBeginDocument{\def\proof{\IEEEproof}}
\fi

\newcommand\ifnull[3]{%
  \ifx\null#1%
    #2%
  \else%
    #3%
  \fi}

\newcommand{\Pade}{Pad\'e\xspace} % easier to write

\newcommand{\word}[1]{\mathrm{#1}}
\newcommand\half{\tfrac 1 2}
\newcommand\defeq{\triangleq}            % define operator
\newcommand\modop{\mathbin\word{mod}}    % mod infix operator (less spacing than \mod)
\newcommand\rem{\mathbin\word{\ rem}}      % mod infix operator (less spacing than \mod)
\newcommand{\mo}{{-1}}                   % minus one as one token (for saving {}'s )
\newcommand{\T}[1]{^{(#1)}}              % Indexes on e.g. sub-results from a algorithm loop
\renewcommand\vec[1]{\bm{#1}}            % Redefine what a vector looks like
\newcommand{\floor}[1]{\lfloor #1 \rfloor}

\newcommand{\diag}{\mathop{\word{diag}}}

\newcommand{\lcm}{\mathop{\word{lcm}}}

\newcommand{\GS}{\word{GS}}
\newcommand{\Pow}{\word{Pow}}

\newcommand{\LM}{{\rm LM}}
\newcommand{\pivs}{{\rm leads}}

\newcommand\F{\mathbb F\xspace} %F
\newcommand\ZZ{\mathbb Z\xspace}

\newcommand\RR{\mathbb R\xspace}
\newcommand{\Code}{\mathcal C}

\newcommand\Oapp{O^{\scriptscriptstyle \sim}\!}
\newcommand\dist{\word{dist}}
\def\fail{\mathsf{fail}}

\makeatletter
\def\easycyrsymbol#1{{\mathord{\mathchoice
  {\mbox{\fontsize\tf@size\z@\usefont{T2A}{cmr}{m}{n}#1}}
  {\mbox{\fontsize\tf@size\z@\usefont{T2A}{cmr}{m}{n}#1}}
  {\mbox{\fontsize\sf@size\z@\usefont{T2A}{cmr}{m}{n}#1}}
  {\mbox{\fontsize\ssf@size\z@\usefont{T2A}{cmr}{m}{n}#1}}
}}}
\makeatother

\newcommand{\rev}[2][\null]{%
\ifnull{#1}{\mathsf{rev}(#2)}{%else
\mathsf{rev}_{#1}(#2)}}

\newcommand\Errs{\mathcal E} % error position set
\newcommand\errs{\epsilon} % number of errors
\newcommand{\ev}[1][]{\word{ev}_{#1}}

\newcommand{\PolyMult}[1]{\mathsf{M}(#1)} % poly multiplication
\newcommand{\MatExp}{\omega} % exponent for matrix multiplication

\algrenewcommand\alglinenumber[1]{{\scriptsize#1}}   %small line numbers
\algrenewcommand\algorithmicrequire{\textbf{Input:}} %instead of "Require"
\algrenewcommand\algorithmicensure{\textbf{Output:}} %instead of "Ensure"
\newcommand{\ass}{\leftarrow}

\title[Power Decoding RS Codes]{Power Decoding Reed--Solomon Codes Up to the Johnson Radius}
\author[Johan Rosenkilde]{}
 
\subjclass{Primary: TODO TODO; Secondary: TODO}

\email{jsrn@jsrn.dk}

\begin{document}
\maketitle

% Enter the first author's name and address:
\centerline{\scshape Johan Rosenkilde}
\medskip
{\footnotesize
% please put the address of the first author
 \centerline{Technical University of Denmark, }
   \centerline{Department of Applied Mathematics and Computer Science}
   \centerline{Denmark}
} % Do not forget to end the {\footnotesize by the sign }

\bigskip

% The name of the associate editor will be entered by an editorial staff
% "Communicated by the associate editor name" is not needed for special issue.
 \centerline{(Communicated by the associate editor name)}

\begin{abstract}
Power decoding, or ``decoding using virtual interleaving'' is a technique for decoding Reed--Solomon codes up to the Sudan radius.
Since the method's inception, it has been an open question if it is possible to use this approach to decode up to the Johnson radius -- the decoding radius of the Guruswami--Sudan algorithm.
In this paper we show that this can be done by incorporating a notion of multiplicities.
As the original Power decoding, the proposed algorithm is a one-pass algorithm: decoding follows immediately from solving a shift-register type equation, which we show can be done in quasi-linear time.
It is a ``partial bounded-distance decoding algorithm'' since it will fail to return a codeword for a few error patterns within its decoding radius; we investigate its failure behaviour theoretically as well as give simulation results.
\ifExtended
This is an extended version where we also show how the method can be made practically faster using a reencoding or a syndrome formulation.
\fi
\end{abstract}

\section{Introduction}

Power decoding was originally proposed by Schmidt, Sidorenko and Bossert for low-rate Reed--Solomon codes (RS)  \cite{schmidt_decoding_2006}.
Using shift-register synthesis techniques, the method can decode as many errors as the Sudan algorithm \cite{sudan_decoding_1997}.
As opposed to Sudan's list decoder, Power decoding is a one-pass algorithm where decoding is realised by solving a simultaneous shift-register problem; however, Power decoding always returns at most one codeword and will for a few error patterns simply fail.
Simulations indicate that this occurs very rarely for random errors\footnote{%
  This behaviour was described as ``probabilistic decoding'' in e.g.~\cite{schmidt_decoding_2006,schmidt_syndrome_2010}.
  However, that term is usually reserved for randomised algorithms such as Las Vegas probabilistic algorithms.
  Power decoding is entirely deterministic, given the input so we prefer the term ``partial decoding''.
}

The Sudan decoder generalises to the Guruswami--Sudan decoder \cite{guruswami_improved_1999} by introducing the multiplicity parameter, improving the decoding radius for all rates up to the Johnson radius \cite{johnson_new_1962}.
Since \cite{schmidt_decoding_2006}, it has been an open question whether it is likewise possible to introduce a ``multiplicity parameter'' into Power decoding and thereby increase the decoding radius up to the Johnson radius.

We settle this question in the affirmative.
The overall behaviour of the obtained decoder is similar to Power decoding:
\begin{enumerate}
  \item The equations are of a generalised shift-register type, and no root-finding as in Guruswami--Sudan is necessary.
  \item The decoding radius becomes almost exactly that of the Guruswami--Sudan decoder (under the same choices of parameters).
  \item When a codeword is returned, it is always a closest codeword.
        The method will fail for a few error patterns whenever one decodes beyond half the minimum distance.
\end{enumerate}
Furthermore, we will show how to realise the decoder efficiently using existing algorithms for solving simultaneous Hermite \Pade approximations: using the algorithms of \cite{gupta_triangular_2012,giorgi_complexity_2003} the complexity becomes $\Oapp(\ell^\MatExp s n)$, where $\MatExp$ is the exponent of matrix multiplication, and $s,\ell$ are the multiplicity, respectively powering parameters of the decoder, and $\Oapp(\cdot)$ is big-$O$ but ignoring $\log(ns\ell)$ factors.
Note that we always have $\ell \geq s$.
The slightly better complexity $\Oapp(\ell^{\MatExp-1} s^2 n)$ can be achieved by relying on \cite{rosenkilde_ne_nielsen_algorithms_2016} (not yet published).
The latter matches the best known complexity for the Guruswami--Sudan algorithm or the Wu list decoder \cite{chowdhury_faster_2015}.

We also investigate the failure behaviour of the proposed method.
Though we do not settle the question of precisely how often Power decoding fails, we make some headway: we prove that the behaviour depends only on the error, and not the sent codeword, and we show that failure can only occur beyond half the minimum distance.
We then give a closed upper bound on the probability for one choice of the algorithm's parameters, $(s,\ell) = (2,3)$.
We also present simulation results that demonstrate a very small probability of failure for larger parameter choices.

Compared to the two existing list decoders, the Guruswami--Sudan and the Wu algorithm, we believe the proposed algorithm is interesting for several reasons.
Firstly, it is more practical since the decoder needs only a single sub-algorithm---of shift-register type---while the other two decoders are more involved.
The Guruswami--Sudan algorithm consists of two steps: interpolation and root-finding.
Interpolation is comparable to the computation in Power decoding (see next section), but root-finding is an additional, non-trivial step.
E.g.~in \cite{ahmed_vlsi_2011}, a hardware implementation of K\"oetter and Vardy's soft-decision variant of Guruswami--Sudan has the critical path in the root-finding unit, and this also uses a significant area of the entire circuit.
In practical applications of hard-decision decoding, one would likely use smaller multiplicities than in \cite{ahmed_vlsi_2011}: this would leave the interpolation unit significantly simpler than that of \cite{ahmed_vlsi_2011} while the root-finding unit would be unchanged, resulting in root-finding occupying relatively even more area and latency.
% Anecdotally, Broadcom also recently revealed interest in using joint decoding of Interleaved Reed--Solomon codes in a communication standard \cite{gottfried_ungerbock_improved_2015}; here it was crucial that the decoder was a one-pass shift-register-based algorithm, based on \cite{schmidt_collaborative_2009}: a method that is computationally very related to Power decoding.

Secondly, Power decoding and the Guruswami--Sudan algorithm can both be adapted to other algebraic constructions, but not always with the same pros and cons; see e.g.~\cite{wachter-zeh_decoding_2012,mohamed_deterministic_2015,nielsen_sub-quadratic_2015,li_decoding_2014} for adaptions of Power decoding.
Case in point, the proposed extension of Power decoding of this paper has already been adapted to improved, quasi-linear time decoding of Interleaved RS codes \cite{puchinger_decoding_2017}.
The previous only known algorithm with the same decoding radius was \cite{cohn_approximate_2012} which can be seen as an adapted Guruswami--Sudan, and which has poor complexity exactly because of the root-finding step.

Lastly, Power decoding decodes beyond half-the-minimum distance without the use of interpolation at all.
This is a rarity in algebraic decoding, and the present paper demonstrates that the Johnson radius can be reached purely linear-algebraically; at least in the presence of random errors.
This sheds new light on decoding of RS codes, and represents an important puzzle piece in relation to the two list-decoding algorithms.

Parts of these results were presented at ACCT-14 \cite{nielsen_power_2014-1}.

\subsection{Related Work}

Power decoding was introduced in \cite{schmidt_decoding_2006, schmidt_syndrome_2010}: for low-rate RS codes, it was shown how one can compute generalised syndromes from ``powering'' the received word, and that these can be used for efficient decoding by solving a multi-sequence shift-register synthesis problem.
One chooses a ``powering degree'' $\ell$: higher $\ell$ yields better decoding radius, but is admissible only for lower-rate codes.
In \cite{schmidt_syndrome_2010}, a bound on the failure probability was given for RS codes over binary extension fields when $\ell=2$, but a general conjecture was given based on simulations results.
The failure behaviour was then further examined in \cite{zeh_unambiguous_2012} and \cite{nielsen_power_2014}, where bounds on the failure probability were obtained over any field for $\ell=2$ and $\ell=3$.
In \cite{nielsen_power_2014}, a reformulation of Power decoding was given based on Gao's decoder \cite{gao_new_2003}, and this was used to show that whether or not Power decoding fails depends only on the error pattern, and not the sent codeword.

The Guruswami--Sudan algorithm \cite{guruswami_improved_1999} is a polynomial-time list-decoding algorithm up to the Johnson radius  $J_{n,k} = n - \sqrt{n(k-1)}$ \cite{johnson_new_1962}.
``List-decoding'' means that the algorithm will return \emph{all} codewords within the decoding radius.
For the algorithm one chooses two parameters $s, \ell \in \ZZ_+$, usually dubbed ``the multiplicity'' respectively ``the list size''.
They satisfy $s \leq \ell$, and they need to grow large for attaining the best decoding radius: for a decoding radius of $J_{n,k} - \varepsilon n$, one needs $s,\ell \in O(1/\varepsilon)$ for any $\varepsilon \in \RR_+$.
See \cite[p. 58]{nielsen_list_2013} for an extreme numerical example with $\varepsilon \approx 1/n^2$.

As noted already in \cite{schmidt_syndrome_2010}, Power decoding is related to Guruswami--Sudan when $s=1$ (also known as ``Sudan decoding'' after \cite{sudan_decoding_1997}): choosing the same value for $\ell$ yields (almost) exactly the same decoding radius.
Computationally, there are more similarities, as noted below.

Guruswami--Sudan consists of two phases, usually dubbed ``Interpolation'' and ``Root-finding'': first, one finds an ``interpolation polynomial'' $Q(y) \in \F[x][y]$, and then one finds $\F[x]$-roots of it.
Both phases have received a tremendous amount of attention with the aim of speeding them up, e.g.~\cite{roth_efficient_2000,alekhnovich_linear_2005,lee_list_2008,beelen_key_2010,zeh_interpolation_2011,cohn_ideal_2010,chowdhury_faster_2015}; see \cite{chowdhury_faster_2015} for an overview on the literature for the Interpolation step.
The best currently known complexities are $\Oapp(\ell^{\omega-1}s^2 n)$ for Interpolation \cite{chowdhury_faster_2015}, and $\Oapp(\ell sn)$ for Root finding \cite{neiger_fast_2017}, if $|\F| \in O(n)$.
Without the use of fast arithmetic, the best known complexities are $O(\ell^3 s^2 n^2)$ for Interpolation \cite{zeh_interpolation_2011}, respectively $O(\ell^2 s^2 n^2)$ for Root finding \cite{roth_efficient_2000}.

One approach for fast Interpolation in Guruswami--Sudan has been to formulate ``Interpolation key equations'', as in \cite{roth_efficient_2000} for the case $s=1$, and \cite{zeh_interpolation_2011} for the general case.
These are shift-register-type equations whose solution result in an interpolation polynomial.
These are related to Power decoding: the generalised syndromes in \cite{roth_efficient_2000} equal those of the original Power decoding \cite{schmidt_decoding_2006}.
However, the two sets of key equations are inherently \emph{different}: the solution to the Power decoding equations yields the error locator, while no clear notion of an error locator is known for the Guruswami--Sudan algorithm.
Similarly, the key equations that we derive in \cref{sec:new_key} bears a resemblance to the equations of \cite{zeh_interpolation_2011}, and it is an interesting question what the algebraic relation between the two approaches is.

The Wu decoding algorithm \cite{wu_new_2008} is an amalgamation between classical key equation decoding \cite{berlekamp_algebraic_1968} and the Guruswami--Sudan: one first attempts half-the-minimum distance decoding using the classical key equation (see the following section).
If this fails, the polynomials computed in the failed attempt are then used to set up a problem solvable by an $\F(x)$-variant of the Guruswami--Sudan algorithm.
One again needs Interpolation and Root-finding sub-algorithms which are similar to, but slightly more involved than, for Guruswami--Sudan; see e.g.~\cite{trifonov_efficient_2012,beelen_rational_2013,chowdhury_faster_2015} for work on these.
The best complexities for these steps equal those of the Guruswami--Sudan algorithm \cite{chowdhury_faster_2015,beelen_rational_2013}.
However, from a practical perspective, the Wu algorithm is slightly more complicated to implement.
The Wu algorithm is also a list-decoding algorithm, and also decodes up to $J_{n,k}$.
Also here one needs to choose parameters $s, \ell$, whose growth relate to the decoding radius as in the Guruswami--Sudan  algorithm \cite{beelen_rational_2013}.

\subsection{Organisation}

In \cref{sec:preliminaries} we give an introduction to the previous key equation-based decoding algorithms: half-the-minimum distance and Power decoding.
In \cref{sec:new_key}, we then derive the new key equations: non-linear relations between known polynomials, revealing the error.
We derive a decoding radius in \cref{sec:radius}, and relate it directly to that of the Guruswami--Sudan algorithm.
Power decoding will fail on certain error patterns within this radius, however, and we investigate this in \cref{sec:failure}.
In \cref{sec:simulation} we give simulation results.
In \cref{sec:pade} we show how to efficiently solve the key equations.
\ifExtended
 In \cref{sec:reencoding} and \cref{sec:syndrome} we investigate re-encoding respectively syndrome reformulations of the proposed key equations, providing practical -- if not asymptotic -- speedups to the decoder.
\else
 In \cref{sec:reencoding} we investigate re-encoding of the proposed key equations, providing practical -- if not asymptotic -- speedups to the decoder.
\fi

The decoding method has been implemented in Sage v8.0 \cite{stein_sagemath_????} and can be downloaded from \url{http://jsrn.dk/code-for-articles}, together with the code for running the simulation.

\section{Preliminaries and Existing Key Equations}
\label{sec:preliminaries}

In complexity discussions, we count arithmetic operations in the field $\F$.
We will use $\omega$ as the exponent for matrix multiplication, i.e.~$2 \leq \omega \leq 3$.
We use $\Oapp(\cdot)$ as big-$O$ but ignoring $\log$-factors.
In a few places we also use $\PolyMult n$ to denote the complexity of multiplying together two polynomials of degree at most $n$; we can trivially use $\PolyMult n \in O(n^2)$ or we can have $\PolyMult n \in \Oapp(n)$, see e.g.~\cite{von_zur_gathen_modern_2012}.

\subsection{GRS codes}

Consider some finite field $\F$.
Choose $n \leq |\F|$ as well as distinct $\alpha_1,\ldots,\alpha_n \in \F$ as well as non-zero (not necessarily distinct) $\beta_1,\ldots,\beta_n \in \F$.
For any $f \in \F[x]$ we write
\[
  \ev(f) = \big( \beta_1f(\alpha_1), \ldots, \beta_nf(\alpha_n) \big) \ .
\]
The $[n,k,d]$ Generalised Reed-Solomon (GRS) code for these parameters is the set
\[
\Code = \big\{ \ev(f) \mid f \in \F[x],~ \deg f < k \big\} \subseteq \F^n \ .
\]
The $\alpha_i$ are called \emph{evaluation points} and the $\beta_i$ \emph{column multipliers}.
$\Code$ has minimum distance $d = n-k+1$ which is the maximal possible according to the Singleton bound.

Consider now that some $\vec c = (c_1,\ldots, c_n)$ was sent with $\vec c = \ev(f)$ for some $f \in \F[x]$, and that $\vec r = (r_1,\ldots, r_n) = \vec c + \vec e$ was the received word with error $\vec e = (e_1,\ldots, e_n)$.
Let $\Errs = \{ i \mid e_i \neq 0 \}$ and $\errs = |\Errs|$.

Note that column multipliers can be ignored in decoding: we simply compute $\vec r' = (r_1/\beta_1,\ldots,r_n/\beta_n) = \vec c' + \vec e'$, where $\vec c'$ is in the code $\Code'$ which has the same evaluation points $\alpha_i$ but where all $\beta_i = 1$.
$\vec e'$ is an error vector with the same number of errors as $\vec e$.
In the remainder of the article, we therefore assume $\beta_i = 1$.

Introduce two essential polynomials, immediately computable by the receiver:
\begin{IEEEeqnarray*}{rCl+rCl}
  \label{eqn:GandR}
  G &=& \prod_{i=1}^n(x-\alpha_i) & R: \deg R < n,~ R(\alpha_i) &=& r_i,\ \ i=1,\ldots,n \ .
\end{IEEEeqnarray*}
$G$ can be pre-computed, while $R$ is computed upon receiving $\vec r$ using Lagrange interpolation.

Key equation decoders revolve around the notion of an error locator $\Lambda$ and error evaluator $\Omega$:
\begin{IEEEeqnarray*}{rCl+rCl}
  \Lambda &=& \prod_{j\in\Errs}(x-\alpha_j)
  & \Omega &=& -\sum_{i \in \Errs} e_i \zeta_i \prod_{j \in \Errs \setminus \{i\}} (x-\alpha_j) \ .
\end{IEEEeqnarray*}
where $\zeta_i = \prod_{j \neq i} (\alpha_i-\alpha_j)^\mo$.
Note that $\errs = \deg \Lambda > \deg \Omega$.

The following simple relation is at the heart of our investigations:
\begin{lemma}
  \label{lem:key_relation}
  $\Lambda(f-R) = \Omega G$ .
\end{lemma}
\begin{proof}
  The closed formula for Lagrange interpolation implies that $f - R = \sum_{i=1}^n -e_i \zeta_i \prod_{j \neq i}(x-\alpha_j)$.
  This directly means
  \[
    \Lambda(f-R) = \Lambda \sum_{i \in \Errs} -e_i \zeta_i \prod_{j \neq i}(x-\alpha_j)
                 = \sum_{i \in \Errs} -e_i \zeta_i \left(\frac \Lambda {x-\alpha_i}\right) G
                 = \Omega G \ .
  \]
\end{proof}

The objects $\vec c, \vec r, \vec e, \Lambda,$ etc. introduced here will be used in the remainder of the article.

\subsection{Classical Key Equations}
\label{ssec:classical_ke}

Let us revisit the key equation implicit in Gao's decoder \cite{gao_new_2003}, which follows directly from \cref{lem:key_relation}:
\begin{equation}
  \label{eqn:key_eq_gao}
  \Lambda R \equiv \Lambda f \mod G \ .
\end{equation}
This is a non-linear equation in the unknowns $\Lambda$ and $f$, and it is not immediately obvious how to build an efficient decoder around it.
The good - and classical - idea is to \emph{linearise} the relation: we replace the sought quantities $\Lambda$ and $\Lambda f$ with unknowns $\lambda$ and $\psi$, both in $\F[x]$, and such that
\[
\lambda R \equiv \psi \mod G \ .
\]
This is now a linear relation with infinitely many solutions.
We further restrict the solutions by requiring
\[
  \deg \lambda + k-1 \geq \deg \psi \ .
\]
Note that this is satisfied if $\lambda$ is replaced by $\Lambda$ and $\psi$ by $\Lambda f$.
Finally, we seek such $\lambda, \psi$ where $\lambda$ is monic and has minimal degree.
The hope is now that $\lambda = \Lambda$ even though we solved for a much weaker relation than \eqref{eqn:key_eq_gao}; effectively, it is therefore the low degree of $(\Lambda R \mod G)$ which is used to solve for $\Lambda$.
Solving such requirements for $\lambda$ and $\psi$ is sometimes known as rational function reconstruction \cite{von_zur_gathen_modern_2012} or \Pade approximation \cite{baker_pade_1996}.
They are easy to solve for in complexity $O(n^2)$ or $\Oapp(n)$, using e.g.~the extended Euclidean algorithm \cite{sugiyama_method_1975,fitzpatrick_key_1995,gao_new_2003}.

It can be shown that whenever $\errs < d/2$ we get $\lambda = \Lambda$ and $\psi = \Lambda f$, see e.g.~\cite{gao_new_2003}.
Then $f = \psi/\lambda$ and decoding is finished.
However, whenever $\errs \geq d/2$, the approach will not work, i.e.~the found $\lambda$ will not equal $\Lambda$.

Whenever $0$ is not an evaluation point, i.e.~$\alpha_i \neq 0$ for all $i$, then the equation can be rewritten to the more classical \emph{syndrome key equation} \cite{berlekamp_algebraic_1968}.
\ifExtended
First some notation: for $p \in \F[x]$, let $\rev[d]p$ denote the \emph{reversal of the coefficients} of $p$ at degree $d$, i.e.~$\rev[d]{p} = x^{d}p(x^\mo)$ for some integer $d \geq \deg p$.
To lighten the notation, we will often omit the $d$-argument when there is an implied upper bound on the degree of the polynomial being reversed; to be precise, note that we then reverse on the \emph{upper bound} on the degree, and not on the actual degree which might happen to be lower.

Introduce $S(x)$ as the power series expansion\footnote{%
  By inserting the explicit Lagrange interpolation formula for $R$, it is easy to see that this definition of the syndrome polynomial corresponds to the classical one, in e.g. \cite[Section 6.2]{roth_introduction_2006}.
} of $\rev R/\rev G$ truncated at $x^{n-k}$.
Then by reversing \cref{lem:key_relation} at degree $\errs+n-1$ we get:
\begin{IEEEeqnarray*}{rCl?r}
  \Lambda R           & =      & \Lambda f - \Omega G                                         & \iff       \\
  \rev[\errs+n-1]{\Lambda R} & =      & \rev[\errs+k-1]{\Lambda f} x^{n-k} - \mbox{$\rev[\errs+n-1]{\Omega G}$}    & \implies \\
  \rev \Lambda \rev R  & \equiv & -\rev \Omega \rev G \mod x^{n-k} \ .
\end{IEEEeqnarray*}
Since $x \nmid \rev G$ this implies the well-known formula:
\begin{equation}
  \label{eqn:key_eq_syndrome}
  \rev \Lambda S  \equiv -\rev \Omega  \mod x^{n-k} \ .
\end{equation}
A (now less obvious) algebraic relation exist between $\rev \Lambda$ and $\rev \Omega$.
To allow for efficient solving, we forget this relation, and replace $\rev \Lambda$ and $-\rev \Omega$ by unknowns $\hat \lambda$ and $\hat \omega$, and solve for the minimal degree $\hat \lambda$ satisfying
\begin{IEEEeqnarray*}{rCl?l}
  \hat \lambda S &\equiv& \hat \omega \mod x^{n-k}  & \textrm{ and} \\
  \deg \hat \lambda &>& \deg \hat \omega \ .
\end{IEEEeqnarray*}
This time the modulus is a power of $x$; solving such an equation for $\hat \lambda$ and $\hat \omega$ is known as \Pade approximations \cite{baker_pade_1996} or a linear feedback shift-register synthesis \cite[Section 6.7]{roth_introduction_2006}.
It can be solved in complexity $O(n^2)$ or $\Oapp(n)$ using either the extended Euclidean algorithm or the Berlekamp--Massey algorithm.

One can again show that this approach will succeed, i.e.~in the end $\hat \lambda = \rev \Lambda$, whenever $\errs \leq \floor{(d-1)/2}$ \cite{berlekamp_algebraic_1968}.
Slightly stronger, one can show that the approach will succeed if and only if the Gao key equation approach succeeds \cite{nielsen_power_2014}.
\else
This has the form
\[
  \hat \lambda S \equiv \hat \omega \mod x^{n-k} \ ,
\]
where $S(x)$ is the \emph{syndrome polynomial} that the receiver computes; we omit the technical rewriting.
The method is successful when $\hat \lambda = x^{\errs}\Lambda(x^\mo)$, i.e.~the coefficient-wise reversal of $\Lambda$, in which case $\Lambda$ reveals the positions of the errors, and one can use e.g.~Forney's formula or erasure decoding to finish decoding.
\fi

\subsection{Simply Powered Key Equations}
\label{ssec:power_ke}

(Simple) Power decoding, or decoding by virtual interleaving \cite{schmidt_syndrome_2010}, is a generalisation of \eqref{eqn:key_eq_gao} where not one but multiple non-linear relations between $\Lambda$ and $f$ are identified, essentially still based on \cref{lem:key_relation}.
The original formulation of \cite{schmidt_syndrome_2010} is based on the classical syndrome key equation, while powering the Gao key equation was described in \cite{nielsen_power_2014}.
We will begin with the latter:
\begin{lemma}[Simply Powered key equations]
  \label{lem:power_simple}
  For any $t \in \ZZ_+$ then
  \[
    \Lambda R^t \equiv \Lambda f^t \mod G \ .
  \]
\end{lemma}
\begin{proof}
  By \cref{lem:key_relation} we have
  \[
    \Lambda f^t = \Lambda \big(R - (R-f)\big)^t = \Lambda R^t + \Lambda (R-f) (\ldots) \equiv \Lambda R^t \mod G \ .
  \]
\end{proof}
Again this gives non-linear relations between $\Lambda$ and $f$.
To solve them efficiently, we choose some $\ell$ and linearise the first $\ell$ of the equations, introducing unknowns $\lambda,\psi_1,\ldots,\psi_\ell \in \F[x]$.
We then solve for $\lambda, \psi_t$ such that $\lambda$ is monic and of minimal degree such that
\begin{IEEEeqnarray*}{rCl?l?l}
  \lambda R^t  & \equiv & \psi_t \mod G \ , & t=1,\ldots, \ell & \textrm{ and } \\
  \deg \lambda & \geq   & \deg \psi_t - t(k-1) \ .
\end{IEEEeqnarray*}
Finally, we hope that the found $\lambda = \Lambda$.
In that case $f = \psi_1/\lambda$ and decoding is finished.

By regarding the linearised problem as a linear system of equations over $\F$, and counting available coefficients versus constraints, one arrives at an expression for the greatest number of errors we should expect to be decodable:
\begin{equation}
  \label{eqn:key_power_gao_radius}
  \errs \leq \tfrac \ell {\ell+1} n - \half \ell(k-1) - \tfrac \ell {\ell+1} \ .
\end{equation}
This argument does not imply that we will necessarily succeed when the bound is satisfied: the constructed system might have spurious ``false solutions'' of degree less than or equal to that of $\Lambda$.
In such rare cases decoding might fail for fewer errors than \eqref{eqn:key_power_gao_radius}.
Bounding the probability that this occurs has proven difficult: we now know upper bounds when $\ell = 2, 3$ \cite{schmidt_syndrome_2010,nielsen_power_2014}, and Schmidt, Sidorenko, and Bossert posed a conjecture, backed by simulation, on the probability in general \cite{schmidt_syndrome_2010}.

From \eqref{eqn:key_power_gao_radius} we can determine the value of $\ell$ that maximise the decoding radius.
Whenever $k/n > 1/3$, one should simply choose $\ell = 1$, i.e.~classical key equation decoding.
Thus simple Power decoding is only useful for low-rate codes.
Note that \eqref{eqn:key_power_gao_radius} is almost the same bound as the Sudan decoding algorithm \cite{sudan_decoding_1997}, which is the Guruswami--Sudan algorithm with multiplicity 1.

\ifExtended
Power decoding was originally described using a syndrome formulation instead of \eqref{eqn:key_power_gao_radius} \cite{schmidt_decoding_2006}:
we restrict ourselves to the case where $0$ is not an evaluation point, and we define $S\T t$ as the power series expansion of $\rev{R\T t}/\rev G$ truncated at $x^{n-t(k-1)-1}$, where $R\T t$ is the unique polynomial of degree less than $n$ such that $R\T t \equiv R^t \mod G$.
Then it follows from \cref{lem:power_simple}, by the same rewriting as in \cref{ssec:classical_ke} \cite{nielsen_power_2014}, that:
\begin{equation}
  \label{eqn:key_power_syndrome}
  \rev \Lambda S\T t \equiv -\rev{\Omega_t} \mod x^{n-t(k-1)-1} ,
\end{equation}
where $\Omega_t$ are certain polynomials of degree at most $\errs-1$ that we omit defining explicitly.
It can be shown using the same rewriting that Power syndrome decoding fails if and only if Power Gao decoding fails \cite{nielsen_power_2014}.
\else
Power decoding was originally given by extending the syndrome formulation mentioned in the previous section.
One now computes $\ell$ ``syndrome polynomials'' $S_1,\ldots,S_\ell$, and seek a solution $\hat \lambda, \hat \omega_1,\ldots,\hat \omega_\ell$ such that
\[
  \hat \lambda S_t \equiv \hat\omega_t \mod x^{n - t(k-1) - 1} \quad t=1,\ldots,\ell \ .
\]
This approach succeeds if and only if the Gao key equation approach succeeds \cite{nielsen_power_2014}.
\fi

For the Gao formulation, the computational problem is sometimes known as ``vector rational function reconstruction'' \cite{olesh_vector_2006}, and for the syndrome formulation as ``simultaneous \Pade approximation'' \cite{baker_pade_1996} or ``multi-sequence shift-register synthesis'' \cite{schmidt_syndrome_2010}.
Iterative algorithms with $O(\ell n^2)$ complexity can be found in \cite{beckermann_uniform_1994,schmidt_syndrome_2010,sidorenko_linear_2011,nielsen_generalised_2013}.
$\Oapp(\ell^\MatExp n)$ algorithms are in \cite{bostan_solving_2008,sidorenko_fast_2014,nielsen_generalised_2013}.
Recently, the improved complexity $\Oapp(\ell^{\MatExp-1} n)$ has been achieved \cite{rosenkilde_ne_nielsen_algorithms_2016}.

\section{New Key Equations}
\label{sec:new_key}

In this section we describe the main result of the paper: a new generalisation of Power decoding where we introduce a second parameter, \emph{the multiplicity}.
The resulting relations will again be non-linear in $\Lambda$ and $f$, and we will employ a linearisation strategy similar to before.

The generalised key equations are described in the following theorem:
\begin{theorem}
  \label{thm:power_key}
  For any $s, \ell \in \ZZ_+$ with $\ell \geq s$, then
  \begin{IEEEeqnarray*}{rCl+l}
    \Lambda^s f^t & =      & \sum_{i=0}^t \big(\Lambda^{s-i}\Omega^i\big)\left(\binom t i R^{t-i} G^i\right)
                  & \textrm{for } t=1, \ldots, s-1 \ , \\
    \Lambda^s f^t & \equiv & \sum_{i=0}^{s-1} \big(\Lambda^{s-i}\Omega^i\big)\left(\binom t i R^{t-i} G^i\right)  \mod G^s
                  & \textrm{for } t=s, \ldots, \ell \ .
  \end{IEEEeqnarray*}
\end{theorem}
\begin{proof}
  We simply rewrite
  \begin{IEEEeqnarray*}{rCl}
    \Lambda^s f^t
    &=& \Lambda^s (R + (f - R))^t \\
    &=& \sum_{i=0}^t \binom t i \Lambda^s (f-R)^i R^{t-i} \ .
  \end{IEEEeqnarray*}
  If $t < s$ then $\Lambda^s (f-R)^i = \Lambda^{s-i} \Omega^i G^i$ for each $i$ by \cref{lem:key_relation}.
  This finishes the first part of the theorem.

  If $t \geq s$ then for $i=s,\ldots,\ell$, the summand equals $\binom t i \Lambda^{i-s}\Omega^sG^sR^{t-i}$ due to \cref{lem:key_relation}, which is $0$ modulo $G^s$.
  Replacing $\Lambda^s (f-R)^i$ by $\Lambda^{s-i} \Omega^i G^i$ for $i < s$ as before gives the sought.
\end{proof}

The above theorem describes $\ell$ equations in the (algebraically related) ``unknowns'' $\Lambda^s, \Lambda^{s-1}\Omega, \ldots, \Lambda \Omega^{s-1}$ as well as $\Lambda^s f, \ldots, \Lambda^s f^\ell$.
These are ``key equations'' in the following sense: inner products of the unknowns $\Lambda^{s-i}\Omega^i$ with vectors of known polynomials (the $\tbinom t i R^{t-i} G^i$) equal the unknowns $\Lambda^s f^t$ modulo $G^s$ -- and hence have surprisingly low degree.

The relations of \cref{thm:power_key} are highly non-linear and solving for $\Lambda$ and $f$ directly would be computationally infeasible.
Instead we \emph{linearise} the relations: derive weaker, linear relations from \cref{thm:power_key} which can be solved efficiently:
\begin{problem}
  \label{prob:ke_linearised}
  Find a vector $(\lambda_1, \ldots \lambda_s, \psi_1, \ldots, \psi_\ell) \in \F[x]^{s+\ell}$ with $\lambda_1$ monic and such that the following requirements are satisfied:
  \begin{IEEEeqnarray*}{r?l?l}
    1a) & \psi_t         =      \sum_{i=0}^{t} \lambda_{i+1} \cdot \left(\binom t i R^{t-i} G^i\right)  \ ,           & \textrm{for } t=1, \ldots, s-1  \\
    1b) & \psi_t         \equiv \sum_{i=0}^{s-1} \lambda_{i+1} \cdot \left(\binom t i R^{t-i} G^i\right) \mod G^s \ , & \textrm{for } t=s, \ldots, \ell \\
    2)  & \deg \lambda_1 \geq   \deg \lambda_{i+1} + i \ ,                                                      & \textrm{for } i=1,\ldots,s-1    \\
    3)  & \deg \lambda_1 \geq   \deg \psi_t - t(k-1) \ ,                                                        & \textrm{for } t=1,\ldots,\ell   \ .
  \end{IEEEeqnarray*}
\end{problem}
Clearly $\vec \Lambda = (\Lambda^s, \Lambda^{s-1}\Omega, \ldots, \Lambda \Omega^{s-1}, \Lambda^s f, \ldots, \Lambda^s f^\ell)$ satisfies the requirements.
The strategy is to find a \emph{minimal solution}, by which we mean that $\deg \lambda_1$ is minimal, and then hope that this solution is actually $\vec \Lambda$.
If that turns out to be the case, decoding can be completed simply by computing $f = \psi_1/\lambda_1$.
\emph{Whether} we can expect that to be the case is addressed in Sections \ref{sec:radius} and \ref{sec:failure}.

The complete decoding algorithm is given as \cref{alg:decoder}, where we assume a solver for \cref{prob:ke_linearised}.
Note that \cref{prob:ke_linearised} could be solved as a series of linear systems in the coefficients of the $\lambda_i$, one system for each guess at $\deg \lambda_1$.
A much more efficient algorithm for solving \cref{prob:ke_linearised} is addressed in \cref{sec:pade}, where we obtain the complexity $\Oapp(s\ell^\omega n)$ for \cref{alg:decoder} (or $\Oapp(s^2\ell^{\omega-1} n)$ relying on the unpublished \cite{rosenkilde_algorithms_2018}).

\begin{algorithm}[t]
  \caption{Efficient Power Decoding with Multiplicities}
  \label{alg:decoder}
  \begin{algorithmic}[1]
    \Require{$\vec r \in \F^n$, $s,\ell \in \ZZ_+$.}
    \Ensure{$\tilde{\vec c} \in \Code$ such that $\dist(\tilde{\vec c}, \vec r)$ is minimal among codewords in $\Code$, or $\fail$\;}
    \State $R \ass $ the Lagrange interpolation polynomial such that $R(\alpha_i) = r_i, i = 1,\ldots,n$.
    \State Compute $\binom t i R^{t-i} G^i \modop G^s$ for $i=1,\ldots,t$ and $t = 1,\ldots,\ell$.
    \State $(\lambda_1,\ldots,\lambda_{s}, \psi_1,\ldots, \psi_\ell) \ass $ a solution to \cref{prob:ke_linearised} such that $\deg \lambda_1$ is minimal.\label{line:decoder:solve_keyeq}
    \State If $\lambda_1$ divides $\psi_1$, let $f \ass \psi_1/\lambda_1$. Otherwise $\fail$.
      \label{line:decoder:compute_f}
    \State If $\dist(\vec r, \ev(f)) = \deg \lambda_1/s$ then return $\ev(f)$. Otherwise $\fail$.
      \label{line:decoder:verify_dist}
  \end{algorithmic}
\end{algorithm}

\begin{remark}
  The shape of the equations of \cref{thm:power_key} bears a striking resemblance to certain approaches for solving the Interpolation phase in the Guruswami--Sudan algorithm: the $\F[x]$ module characterisation as in \cite{lee_list_2009,beelen_key_2010}, and the (intermediate) Interpolation key equations as in \cite[Eqn. (31)]{zeh_interpolation_2011}.
  However, the Guruswami--Sudan algorithm has, a priori, nothing to do with the error locator, and the true connection between the two sets of key equations is unclear.
  For instance, it is not known if one can easily obtain the error locator from a Guruswami--Sudan interpolation polynomial or vice versa.
\end{remark}

\begin{remark}
  The original Power decoding can be described by analogy with decoding of certain Interleaved RS codes \cite{schmidt_syndrome_2010}.
  It would be interesting to find a similar analogue for the key equations of \cref{thm:power_key}.
\end{remark}

\section{Decoding Radius}
\label{sec:radius}

\newcommand\red[1]{\check{#1}}

We will now discuss how many errors \cref{alg:decoder} will usually be able to correct.
When calling this a ``decoding radius'' we need to be wary: indeed, the method \emph{will} fail for certain received words whenever the number of errors is at least $d/2$, and this is unavoidable since it is a unique decoding algorithm.
Therefore, ``decoding radius'' really involves two parts: 1) how many errors should we at most expect to be able to correct; and 2) what is the probability that we will fail when the number of errors is at most this.
In this section we will answer the first of these questions, and turn to the latter in \cref{sec:failure}.

The decoding radius upper bound that we will derive is based on linear algebra: when the number of errors $\errs$ is large enough, then solutions to \cref{prob:ke_linearised} that are smaller than the sought $\vec \Lambda$ will appear.

\begin{proposition}
  \label{prop:degree_upper_bound}
  Consider a received word $\vec r$ and the corresponding instance of \cref{prob:ke_linearised}.
  There is a vector $\vec v = (\red \lambda_1,\ldots,\red \lambda_s, \red\psi_1, \ldots \red\psi_\ell)$ satisfying Items 1a and 1b of \cref{prob:ke_linearised} as well as:
  \begin{IEEEeqnarray*}{r?l?l}
    2')  & s\tau_\Pow(s, \ell) \geq   \deg \lambda_{i+1} + i \ ,                                                      & \textrm{for } i=0,\ldots,s-1    \\
    3')  & s\tau_\Pow(s, \ell) \geq   \deg \psi_t - t(k-1) \ ,                                                        & \textrm{for } t=1,\ldots,\ell   \ .
  \end{IEEEeqnarray*}
  where
  \begin{equation}
    \label{eqn:degree_upper_bound}
    \tau_\Pow(s,\ell) = \frac {2\ell - s + 1}{2(\ell+1)}n - \frac \ell {2s} (k-1) - \frac \ell {s(\ell+1)}  \ .
  \end{equation}
  If $\errs > \tau_\Pow(s, \ell)$ then $\deg \Lambda > \deg \red\lambda_1/s$.
\end{proposition}
\begin{proof}
  Satisfying Items 1a, 1b of \cref{prob:ke_linearised} as well as Items 2' and Items 3' above is a homogeneous linear set of restrictions in the coefficients of the $\lambda_i$: the linear combinations on the right-hand side of Items 1a and 1b should have bounded degree, either directly or reduced modulo $G^s$.
  If there are more coefficients than constraints, there will be a solution.

  Let us write $\tau = \tau_\Pow(s,\ell)$ for brevity; we will derive that if $\tau$ satisfies \eqref{eqn:degree_upper_bound}, then there will be a solution to the homogeneous system.
  For every $t = 1,\ldots,s-1$, Item 1a imposes $C_t$ constraints, where:
  \begin{IEEEeqnarray*}{rCl}
    C_t &=& \deg(\textrm{rhs}) - (s\tau + t(k-1)) \\
        &=& \max_{i=0,\ldots,s-1}( s\tau - i + (n-1)(t-i) + in) - (s\tau + t(k-1)) \\
        &=& tn - t - t(k-1) = t(n-1-(k-1)) \ .
  \end{IEEEeqnarray*}
  For Item 1b, then $\psi_t$ has bounded degree modulo $G^s$, so this gives for $t = s,\ldots,\ell$:
  \begin{IEEEeqnarray*}{rCl}
    C_t &=& sn-1 - (s\tau + t(k-1))
  \end{IEEEeqnarray*}
  We thus have a total number of constraints:
  \begin{IEEEeqnarray*}{rCl}
    \sum_{t=1}^\ell C_t
    &=& \sum_{t=1}^{s-1}\big(t(n-1-(k-1))\big) + \sum_{t=s}^\ell\big(sn - s\tau - t(k-1) - 1 \big) \\
    &=& \half (2\ell - s + 1)sn - \tbinom{\ell+1} 2 (k-1) - (\ell-s+1)s \tau - (\tbinom s 2 + \ell-s+1)
  \end{IEEEeqnarray*}
  The total number of coefficients in $\lambda_1,\ldots,\lambda_s$ is:
  \begin{IEEEeqnarray*}{rCl}
    K = \sum_{i=0}^{s-1}(s \tau - i + 1)
    &=& s^2 \tau - \tbinom s 2 + s
  \end{IEEEeqnarray*}
  The condition for a guaranteed solution is then $K > \sum_{t=1,\ldots,\ell} C_t$, i.e.:
  \begin{IEEEeqnarray*}{rCl}
    (\ell+1)s\tau > \half (2\ell - s + 1)sn - \tbinom{\ell+1} 2 (k-1) - \ell-1 \ .
  \end{IEEEeqnarray*}
  Thus, there must be a solution satisfying Items 1a, 1b, 2' and 3' for $\tau$ satisfying:
  \begin{IEEEeqnarray*}{rCl}
    (\ell+1)s\tau = \half (2\ell - s + 1)sn - \tbinom{\ell+1} 2 (k-1) - \ell \ .
  \end{IEEEeqnarray*}
\end{proof}

The solution $\red\lambda_1,\ldots,\red\lambda_s$ guaranteed by \cref{prop:degree_upper_bound} will not necessarily solve \cref{prob:ke_linearised}: it might e.g.~be that $\deg \red\lambda_1 < \deg \red\lambda_2 + 1 \leq s \tau_\Pow(s,\ell)$.
However, it is natural to suspect that, once there are solutions to the system of \cref{prop:degree_upper_bound}, there will be solutions with $\deg \red \lambda_1 = s \tau_\Pow(s, \ell)$, and such solutions will necessarily also solve \cref{prob:ke_linearised}.
The minimal solution to \cref{prob:ke_linearised} will in such cases not be $\vec\Lambda$ that we are looking for.
Therefore, we might \emph{expect} to fail, whenever $\errs > \tau_\Pow(s, \ell)$.
This intuition is completely backed by simulation, see \cref{sec:simulation}: with high probability, decoding seems to fail if $\errs > \tau_\Pow(s, \ell)$, but for a few error patterns it does succeed after all.
We will therefore regard $\tau_\Pow(s, \ell)$ as the decoding radius of \cref{alg:decoder}.

The expression $\tau_\Pow(s,\ell)$ turns out to related to something very well known:
\begin{corollary}
  \label{cor:power_vs_gs}
  Denote the maximal decoding radius of the Guruswami--Sudan algorithm on $\Code$ with multiplicity $s$ and list size $\ell$ by $\tau_\GS(s,\ell)$.
  Then
  \[
    \tau_\GS(s,\ell) = \frac {2\ell -s + 1}{2(\ell+1)} n - \frac {\ell}{2s}(k-1) = \tau_\Pow(s,\ell) + \frac{\ell}{s(\ell+1)} \ .
  \]
  (see e.g.~\cite[Lemma 9.5]{roth_introduction_2006}).
\end{corollary}

Taken over all $s$ and $\ell$, the decoding radius of Guruswami--Sudan describes a curve $J(n,d) = n - \sqrt{n(n-d)}$, often called the Johnson radius after \cite{johnson_new_1962}.
For any integer $\tau < J(n,d)$ there exists infinitely many choices of $s, \ell$ such that $\tau = \floor{\tau_\GS(s,\ell)}$.
Thus, by \cref{cor:power_vs_gs}, Power decoding is similarly bounded by the Johnson radius (for $s, \ell \rightarrow \infty$ then $\tau_\GS(s,\ell) - \tau_\Pow(s,\ell) \rightarrow 0$).
The corollary even allows us to use closed-form expressions for small $s$ and $\ell$ already analysed for the Guruswami--Sudan algorithm:

\begin{proposition}
  \label{prop:parameter_choices}
  Given a decoding radius $\tau < J(n,d) = n - \sqrt{n(n-d)}$, let $\tilde\tau := \tau + \tfrac 1 {s(\tau)}$.
  As long as $\tilde\tau < J(n,d)$ then $\tau_\Pow\big(s(\tilde\tau),\ell(\tilde\tau)\big) \geq \tau$, where
  \begin{IEEEeqnarray*}{rCl+rCl}
    s(\tau)    &=& \lfloor s_{\min}(\tau) + 1 \rfloor
      &
    \ell(\tau) &=& \left \lfloor \tfrac{n-\tau}{k-1} \cdot s(\tau) + \half - \tfrac {\sqrt{D(\tau)}}{k-1} \right \rfloor \ ,
  \end{IEEEeqnarray*}
  % \vspace*{-2em}
  \begin{IEEEeqnarray*}{rCl}
    s_{\min}(\tau) & = & \frac {\tau (k-1)}{(n-\tau)^2-n(k-1)} \\
    D(\tau)       & = & \big(s(\tau) - s_{\min}(\tau)\big)\cdot \big((n-\tau)^2 - n(k-1)\big) \cdot s(\tau) + \tfrac {(k-1)^2} 4
  \end{IEEEeqnarray*}
\end{proposition}
\begin{proof}
  Since $\tilde\tau < J(n,d)$, it is a valid decoding radius for the Guruswami--Sudan algorithm, and so by \protect{\cite[p. 53]{nielsen_list_2013}}, then $\tau_\GS\big(s(\tilde\tau),\ell(\tilde\tau)\big) \geq \tilde\tau$.
  Therefore \cref{cor:power_vs_gs} gives us $\tau_\Pow(s,\ell) \geq \tau + \frac 1 {s(\tau)} - \frac {\ell(\tilde\tau)} {s(\tilde\tau)(\ell(\tilde\tau)+1)}$, so we are done if $s(\tau) \leq s(\tilde\tau)$.
  But $s_{\min}(\tau)$ is monotonically increasing for $0 < \tau < J(n,d)$ so $s(\tau)$ is non-decreasing,
\end{proof}

\begin{remark}
  We remark that the condition of \cref{prop:parameter_choices} that $\tilde\tau < J(n,d)$ seems almost always to be verified: for $n < 100$ an exhaustive search found only 50 choices of the triple $(n,k,\tau)$ for which it was not verified, and in 48 of these cases $n = k+3$.
  As an example of the tightness of the closed expressions, consider the large parameters $[n,k] = [243320, 131155]$: here the list size of \cref{prop:parameter_choices} never exceeds the minimal possible by more than 1 for all possibly decoding radii.
\end{remark}

\section{Failure Behaviour}
\label{sec:failure}

\def\hc{\hat{\vec c}}
\def\hR{\hat R}
\def\hf{\hat f}
\def\hpsi{\hat \psi}
\def\lambdar{\check \lambda}
\def\psir{{\check \psi}}

We will move on to investigate how Power decoding fails when at most $\tau_\Pow(s,\ell)$ errors occur.
There are two ways in which \cref{alg:decoder} can give an unwanted answer: firstly, the algorithm can return $\fail$; or secondly, the algorithm can return a different codeword than the sent one.
For a specific sent codeword $\vec c$ and received word $\vec r$, we say that Power decoding \emph{fails} if one of the two following conditions are satisfied:
\begin{enumerate}
  \item \cref{alg:decoder} returns $\fail$.
  \item There exists $\vec c' \in \Code$, $\vec c' \neq \vec c$, and such that $\dist(\vec r, \vec c') \leq \dist(\vec r, \vec c)$.
\end{enumerate}
Recall that when \cref{alg:decoder} does not return $\fail$, it always returns a codeword of minimal distance to the received.
So if neither of the above conditions are satisfied, \cref{alg:decoder} returns the correct answer.
Contrarily, if only item 2 above is satisfied and $\dist(\vec r, \vec c') = \dist(\vec r, \vec c)$, then $\vec c$ might still be correctly returned.
However, it is much more likely that the found solution to the key equation in \cref{line:decoder:solve_keyeq} will be some mix of the solutions corresponding to the two errors $\vec r - \vec c$  and $\vec r - \vec c'$, in which case decoding will fail.
For the sake of a cleaner definition, we therefore consider this possibility as a failure as well.

We will begin with showing that the error vector alone determines whether the method succeeds.
This drastically simplifies further examinations on the failure behaviour.
It allows us first to show the---quite expected---property that the method never fails when fewer than $d/2$ errors occur.
Secondly, it allows us to give a closed upper bound on the failure probability when $(s,\ell) = (2,3)$.
Lastly, we discuss the relation between Power decoding failing and having multiple close codewords to the received word.

\begin{proposition}
  \label{prop:power_gao_inv}
  The success of Power decoding $\vec r = \vec c + \vec e$ depends only on the error $\vec e$.
\end{proposition}
\begin{proof}
  It suffices to show that if Power decoding fails for $\vec r$ as received word, then Power decoding also fails for $\vec r + \hc$ where $\hc$ is any codeword.
  If decoding fails on input $\vec r$ this is because there exist $\lambda_1,\ldots, \lambda_s, \psi_1,\ldots,\psi_\ell \in \F[x]$ which solve \cref{prob:ke_linearised}, and where $\lambda_1 \neq \Lambda^s$ and $\deg \lambda_1 \leq \deg \Lambda^s$.
  Assume this is the case.
  Let $\hR$ be the Lagrange interpolant corresponding to $\vec r + \hc$ as received word, i.e.~$\hR = R + \hf$ where $\hf = \ev^\mo(\hc)$ and $\deg \hf < k$.
  We will show that there exist $\hpsi_1,\ldots,\hpsi_\ell \in \F[x]$ such that the $\lambda_i, \hpsi_t$ form a solution to \cref{prob:ke_linearised} for $\hR$ in place of $R$.
  Therefore, Power decoding will also fail for $\vec r + \hc$ as received word.

  Consider for $t=1,\ldots,\ell$ the following expansion:
  \begin{IEEEeqnarray*}{l}
    \sum_{i=0}^{\min(t, s-1)} \lambda_{i+1} \cdot \left(\binom t i \hR^{t-i} G^i\right) \\
    = \sum_{i=0}^{\min(t, s-1)} \lambda_{i+1} \binom t i \left( \sum_{h=0}^{t-i} \binom{t-i}{h} R^{t-i-h} \hf^h \right) G^i \\
    = \sum_{h=0}^{t} \hf^h \sum_{i=0}^{\min(t-h,s-1)} \lambda_{i+1} \binom t i \binom{t-i}{h} R^{t-i-h} G^i \ . \\
  \end{IEEEeqnarray*}
  Note now that $\binom t i \binom {t-i} h = \binom t h \binom {t-h} i$.
  Therefore, the above equals
  \begin{IEEEeqnarray*}{ll}
    \sum_{h=0}^{t} \binom t h \hf^h \sum_{i=0}^{\min(t-h,s-1)} \lambda_{i+1} \binom{t-h}{i} R^{t-i-h} G^i \\
    \equiv \sum_{h=0}^{t} \binom t h \hf^h \psi_{t-h} \ ,
  \end{IEEEeqnarray*}
  where we by ``$\equiv$'' mean $=$ when $t < s$ and congruent modulo $G^s$ when $t \geq s$.
  We set $\hpsi_t$ as the last expression above.
  By hypothesis, $\deg \psi_{t-h} - (t-h)(k-1) < \deg \lambda_1$.
  Since $\deg \hf < k$ we therefore get $\hpsi_t - t(k-1) < \deg \lambda$.

  This means the $\lambda_i,\psi_t$ indeed form a solution to \cref{prob:ke_linearised} for $\hR$, as we set out to prove.

  The proved implication can immediately be applied in the other direction since $-\vec\hc$ is a codeword, showing the bi-implication.
\end{proof}

We now prove that Power decoding always succeeds in half-the-minimum distance decoding.
The proof is surprisingly technical since we need to keep a handle on all the key equations simultaneously.

\begin{proposition}
  \label{prop:mindist}
  If fewer than $d/2$ errors occur, then Power decoding succeeds.
\end{proposition}
\begin{proof}
  By \cref{prop:power_gao_inv}, we can assume that $\vec 0$ was sent.
  By \cref{lem:key_relation} we then have $R = -\Omega \Upsilon$, where $\Upsilon = G/\Lambda$.

  Assume contrary to the proposition that Power decoding has failed.
  That means there exists $(\lambda_1,\ldots, \lambda_s, \psi_1,\ldots,\psi_\ell)$ which solve \cref{prob:ke_linearised}, and where $\lambda_1 \neq \Lambda^s$ and $\deg \lambda_1 \leq \deg \Lambda^s$.
  We will inductively establish $P(t)$ for $t=0,\ldots,s-1$, where $P(t)$ is the assertion
  \begin{equation*}
    P(t):\quad \Lambda^{t+1-i} \mid \lambda_{i+1} \textrm{ and } \psi_{s-i} = 0 \, \textrm{ for } i=0,\ldots,t \ .
  \end{equation*}
  For $t=s-1$, $P(t)$ implies $\Lambda^s \mid \lambda_1$, which contradicts the minimality of $\lambda_1$, finishing the proof.

  For the case $P(0)$, we need to prove that $\Lambda \mid \lambda_1$ and $\psi_s = 0$.
  Consider the $s$'th key equation of \cref{prob:ke_linearised} which is satisfied by the $\lambda_{i+1}$ and $\psi_s$:
  \begin{equation}
    \label{eqn:mindist:key_s}
    \psi_s \equiv \sum_{i=0}^{s-1} \binom s i \lambda_{i+1} R^{s-i} G^i \mod G^s \ .
  \end{equation}
  $\Upsilon^s$ divides each term of the summand, as well as the modulus $G^s$, and so it must divide $\psi_s$.
  However, we have
  \[
    \deg \psi_s \leq \deg \lambda_1 + s(k-1) \leq s\errs + s(k-1) < s(n-\errs) \ ,
  \]
  where the last inequality holds since $2\errs < n-k+1$.
  Thus $\psi_s = 0$.

  Returning to \eqref{eqn:mindist:key_s}, we can then conclude $\Lambda \mid \lambda_1 R^s$ , since $\Lambda$ divides every other term in the sum as well as the modulus.
  This implies $\Lambda \mid \lambda_1$ since $\gcd(\Lambda, R) = 1$.

  For the inductive step, assuming $P(t-1)$ we will prove $P(t)$ for $1 \leq t < s$.
  Consider now the $(s-t)$'th key equation, i.e.~
  \begin{IEEEeqnarray*}{rCl}
    \psi_{s-t} = \sum_{i=0}^{s-t} \binom {s-t} i \lambda_{i+1} R^{s-t-i} G^i \ .
  \end{IEEEeqnarray*}
  Similar to before, $\Upsilon^{s-t}$ divides every term of the sum, so it divides $\psi_{s-t}$.
  By $P(t-1)$ then $\Lambda^{t-i} \mid \lambda_{i+1}$ for $i = 0,\ldots, t-1$, and therefore $\Lambda^t \mid \lambda_{i+1} R^{s-t-i}G^i$.
  This implies $\Lambda^t \mid \psi_{s-t}$ and hence $\Upsilon^{s-t} \Lambda^t \mid \psi_{s-t}$.
  But now we have
  \[
    \deg \psi_{s-t} \leq \deg \lambda_1 + (s-t)(k-1) \leq s\errs + (s-t)(k-1) < (s-t)(n-\errs) + t\errs \ ,
  \]
  which means $\psi_{s-t} = 0$.

  It remains to show that $\Lambda^{t+1-i} \mid \lambda_{i+1}$ for $i=0,\ldots,t$.
  For $j = 1,\ldots, t$, multiply the $(s-j)$'th key equation with $R^j$ and relax it to a congruence modulo $G^s$.
  We obtain $t+1$ homogeneous linear equations in $\lambda_{i+1} R^{s-i}G^i$ of the form:
  \[
    0 \equiv \sum_{i=0}^{\min(s-1,s-j)} \binom {s-j} i (\lambda_{i+1} R^{s-i} G^i) \mod G^s \ , \quad j=0,\ldots,t \ .
  \]
  Subtracting the $j$th equation from the $(j-1)$st for $j=1,\ldots,t$, we eliminate $\lambda_1$ and get
  \[
    0 \equiv \sum_{i=1}^{s-1} \binom {s-j} {i-1} (\lambda_{i+1} R^{s-i} G^i) \mod G^s \ , \quad j=1,\ldots,t \ .
  \]
  This can be continued to get a series of equation systems, that is, for $t' = 1,\ldots,t$, we have a system:
  \[
    0 \equiv \sum_{i=t'}^{s-1} \binom {s-j} {i-t'} (\lambda_{i+1} R^{s-i} G^i) \mod G^s \ , \quad j=t',\ldots,t \ .
  \]
  For $t'=t$, the system (which is one equation) implies that $\Lambda^{t+1} \mid \lambda_{t+1}R^{s-t}G^t$ since $\Lambda^{t+1}$ divides all the sum's other terms and the modulus, and this implies $\Lambda \mid \lambda_{t+1}$.
  We can now go to the $t' = t-1$ system and regard any of the two equations, and we conclude similarly that $\Lambda^{t+1} \mid \lambda_{t} R^{s-t+1}G^{t-1}$ since $\Lambda^{t+1}$ now is seen to divide all other terms of the sum as well as the modulus.
  This implies $\Lambda^2 \mid \lambda_{t}$.
  Continuing with decreasing $t'$ we can iteratively conclude $\Lambda^{t+1-t'} \mid \lambda_{t'+1}$.

  This finishes the induction step, establishing $P(t)$ for $t=0,\ldots,s-1$.
  As mentioned, this implies a contradiction, finishing the proof.
\end{proof}

We are now in a position to bound the probability that Power decoding fails if errors of a given weight are drawn uniformly at random, for the case $(s,\ell)=(2,3)$.
Note that by \cref{prop:degree_upper_bound}, then
\[
  \tau_\Pow(2,3) = \frac 5 8 n - \frac 3 4 (k-1) - \frac 3 8 \ ,
\]
so these parameters allow the decoder to improve upon both half-the-minimum distance and the original Power decoding whenever the rate is between $1/6$ and $1/2$, for long enough codes.

\begin{proposition}
  \label{prop:failure_prob_23}
  Let $q = |\F|$.
  Whenever $d/2 \leq \errs < \tau_\Pow(2,3)$, the probability that Power decoding fails is upper bounded by
  \[
  \left\{
  \begin{array}{ll}
    \displaystyle
     4 \big(q^{-8}\big)^{(\tau_\Pow(2,3)-\errs) - (0.29 \errs/\log q-1/4)} & \textrm{when } \errs \geq \frac 3 5 n - \frac 4 5 (k-1)  \\[3pt]
    \displaystyle
     4 q^{-d/5 + 1 + 1.61\errs/\log q}  & \textrm{when }  \errs < \frac 3 5 n - \frac 4 5 (k-1)
  \end{array}
  \right .
  \ .
  \]
\end{proposition}
\begin{proof}
  By \cref{prop:power_gao_inv}, we can consider the probability over the choice of error vector, and simply bound the failure probability when the sent codeword was $\vec 0$.
  Since we know by \cref{prop:mindist} that the failure probability is zero when $\errs < d/2$, then we can also assume $\errs \geq d/2$.

  Fix now the number of errors $\errs$ and error positions $\Errs$, implying a specific $\Lambda$.
  For a given error $\vec e = \vec r$ with these non-zero positions, we will call $\vec r$, or $R$, ``bad'' if for $R$ there exist $\lambda_i, \psi_t$ solving \cref{prob:ke_linearised} and such that $\lambda_1 \neq \Lambda^s$ while $\deg \lambda_1 \leq \deg \Lambda^s$.
  Consequently, Power decoding fails only for bad error-values.
  Denote by $S_\Lambda \subset \F[x]$ the set of bad $R$.
  We will give an upper bound $N$ on the size of $S_\Lambda$ and so $N/(q-1)^{\errs}$ bounds the probability that for the fixed error positions, Power decoding fails (since for each position, we have $q-1$ choices of an error value).
  $N$ will turn out to be independent of the choice of $\Lambda$, and thus $N/(q-1)^\errs$ is a bound on the probability that Power decoding fails for any error of weight $\errs$.

  By assumption, the following equations are satisfied:
  \begin{IEEEeqnarray*}{rCl;l;l;l}
    \psi_1 & =      & \lambda_1 R   & + & \phantom{2}\lambda_2 G \ ,               \\
    \psi_2 & \equiv & \lambda_1 R^2 & + & 2 \lambda_2 R G    & \mod G^2 \ ,  \\
    \psi_3 & \equiv & \lambda_1 R^3 & + & 3 \lambda_2 R^2 G  & \mod G^2 \ .
  \end{IEEEeqnarray*}
  Since $R(\alpha_i) = 0$ whenever $i \notin \Errs$, then $\Upsilon \mid R$ where $\Upsilon = G/\Lambda$.
  Thus the above implies $\Upsilon \mid \psi_1$ and $\Upsilon^2 \mid \psi_t$ for $t=2,3$.
  Furthermore, letting $g \defeq \gcd(\lambda_1, \Lambda)$, we can conclude that $g = \gcd(\psi_t, \Lambda)$ for all $t$.
  The regular form of the above three equations allows eliminating $\lambda_1$ and obtain:
  \begin{IEEEeqnarray*}{rCl;l}
    \psi_2 - R\psi_1 &\equiv& \lambda_2 R G    & \mod G^2 \ , \\
    \psi_3 - R\psi_2 &\equiv& \lambda_2 R^2 G  & \mod G^2 \ .
  \end{IEEEeqnarray*}
  From this we first note that $G \mid (\psi_2 - R\psi_1)$.
  We will use this fact momentarily.
  With the two above equations we continue to eliminate $\lambda_2$ and rewrite:
  \begin{IEEEeqnarray*}{rCl;l+r}
    \psi_3 - R\psi_2 - R(\psi_2 - R\psi_1)      & \equiv & 0        & \mod G^2 & \iff     \\
    R^2\psi_1 - 2R\psi_2 + \psi_3               & \equiv & 0        & \mod G^2 & \implies \\
    R^2\psi_1^2 - 2R\psi_1\psi_2 + \psi_1\psi_3 & \equiv & 0        & \mod G^2 & \iff     \\
    (R\psi_1 - \psi_2)^2 + \psi_1\psi_3         & \equiv & \psi_2^2 & \mod G^2 \ .
  \end{IEEEeqnarray*}
  But we concluded just before that $G \mid (R\psi_1 - \psi_2)$ so $(R\psi_1 - \psi_2)^2 \equiv 0 \mod G^2$.
  This leaves the simple relation
  \begin{IEEEeqnarray}{rCl+l}
    \psi_2^2          & \equiv & \psi_1 \psi_3 \mod G^2 \notag & \iff \\
    \psir_2^2\Upsilon & \equiv & \psir_1 \psir_3 \mod \Lambda^2 \ , \label{eqn:failure_prob_psir_relation}
  \end{IEEEeqnarray}
  where $\psir_t \defeq \psi_t/\Upsilon^{\min(2,t)}$, and is a polynomial by our earlier observations.
  Thus, whenever $R$ is bad, there is a triple $(\psir_1,\psir_2,\psir_3) \in \F[x]$ satisfying the above relation as well as
  \begin{equation}
    \label{eqn:failure_prob_psir_deg}
    \deg \psir_t \leq d_t \defeq 2\errs + t(k-1) - \min(2,t)(n-\errs) \ .
  \end{equation}
  We will count the number of such triples momentarily.
  However, to thusly bound the number of bad error values, we have to determine how many different $R$ could have the same triple.
  Recall that determining $R$ up to congruence modulo $\Lambda$ suffices, since this determines the error values.
  However, by our previous observation we have
  \begin{IEEEeqnarray*}{rCl+l}
    R\psi_1 &\equiv& \psi_2 \mod G & \iff \\
    R\psir_1 &\equiv& \psir_2\Upsilon \mod \Lambda & \iff \\
    R &\equiv& (\psir_2\Upsilon/g)(\psir_1/g)^\mo \mod \Lambda/g \ .
  \end{IEEEeqnarray*}
  This means that for a given triple $(\psir_t)_t$, having $\gcd(\psir_t, \Lambda) = g$, there can be at most $q^{\deg g}$ possible choices of $R$.

  To bound the number of bad error values $N$ for this given $\Lambda$, we will therefore perform a weighted count of all triples satisfying  \eqref{eqn:failure_prob_psir_relation} and \eqref{eqn:failure_prob_psir_deg}, where a triple is counted with weight $q^{\deg g}$, where $g$ is a divisor of $\Lambda$ dividing all the $\psir_t$:

  \begin{IEEEeqnarray*}{rCl}
  N
  &\leq& \sum_{g \mid \Lambda} q^{\deg g}
    \Big| \Big\{ (\psir_t)_t \in \F[x]^3 \ \Big|\
      g \mid \psir_t,\ \deg \psir_t \leq d_t,\ \Upsilon \psir_2^2 \equiv \psir_1 \psir_3 \mod \Lambda^2 \Big\} \Big| \\
  &=& \sum_{g \mid \Lambda} q^{\deg g}
    \Big| \Big\{ (\check \psir_t)_t \in \F[x]^3 \ \Big|\
      \deg \check \psir_t \leq d_t - \deg g,\ \Upsilon \check \psir_2^2 \equiv \check \psir_1 \check \psir_3 \mod (\Lambda/g)^2 \Big\} \Big|
      \ .
  \end{IEEEeqnarray*}

  Let $T_g$ be the set inside the last sum.
  We use \cref{lem:abstract_congruence_bound} (see below) to upper bound $|T_g|$, for any choice of $g$:
  setting $A = (\Lambda/g)^2$, $B = \Upsilon$ and $K_t = d_t - \deg g$ in that lemma, we get
  \begin{IEEEeqnarray*}{rCl}
    |T_g| &\leq& 2^{\gamma + 2\errs - 2\deg g} q^{4\errs - (2n-2(k-1)) + 1 + \max(0, \gamma) - \deg g} \ ,
  \end{IEEEeqnarray*}
  where $\gamma = 5\errs - (3n - 4(k-1))$.
  This is only dependent on the \emph{degree} of $g$.
  For each choice of $\deg g$, we can select $g$ in $\binom \errs {\deg g}$ ways since $g \mid \Lambda$ and $\Lambda$ splits into $\errs$ linear factors.
  This gives
  \begin{IEEEeqnarray}{rCl}
    \label{eqn:more_precise_bound}
    N &\leq& q^{4\errs + 2n + 2(k-1) + 1 + \max(0,\gamma)} 2^\gamma \sum_{t=0}^\errs \binom \errs t 4^{\errs - t}
  \end{IEEEeqnarray}
  This can be simplified at a small loss of tightness.
  Firstly $\sum_{t=0}^\errs \binom \errs t 4^{\errs - t} = (4+1)^\errs$.
  For the case $\gamma \geq 0$, we rewrite into
  \begin{IEEEeqnarray*}{rCl}
    N &\leq& q^{\errs + 8(\errs - \tau_\Pow(2,3)) - 2} 2^\gamma 5^\errs
    \qquad (\gamma \geq 0) ,
\end{IEEEeqnarray*}
  since $8\errs - (5n - 6(k-1) - 3) = 8(\errs - \tau_\Pow(2,3))$.
  Now $\gamma \leq \errs$, as can be seen as follows:
  since $\errs < \tau_\Pow(2,3)$ then $4\errs < \frac 5 2 n - 3(k-1)$.
  Inserting in the expression for $\gamma$, we get that $\gamma \leq \errs - (n/2 - (k-1))$.
  But since we assumed $d/2 \leq \tau_\Pow(2,3)$ then $k-1 \leq n/2$ which means $\gamma \leq \errs$.
  Therefore:
  \begin{IEEEeqnarray*}{rCl}
    N &\leq& q^{\errs + 8(\errs - \tau_\Pow(2,3)) - 2} 10^\errs \qquad (\gamma \geq 0) \ .
  \end{IEEEeqnarray*}
  For the case $\gamma < 0$, we instead get
  \begin{IEEEeqnarray*}{rCl}
    N &\leq& q^{4\errs - 2d + 1} 5^\errs  \qquad (\gamma <  0) \ .
  \end{IEEEeqnarray*}
  Since $\gamma < 0$ then $\errs < \tfrac 3 5 n - \tfrac 4 5 (k-1) < \tfrac 3 5 d$, and so $3\errs - 2d < -d/5$, which gives
  \begin{IEEEeqnarray*}{rCl}
    N &\leq& q^{\errs - d/5 + 1} 5^\errs  \qquad (\gamma <  0) \ .
  \end{IEEEeqnarray*}
  As previously described $N/(q-1)^\errs$ then becomes a bound on the probability of decoding failure.
  The term $(\frac q {q-1})^\errs$ then appears, but $(\frac q {q-1})^\errs \leq (\frac q {q-1})^q \leq 2^2$.
  Finally, $10^\errs = (q^{-8})^{-\errs \log 10/(8\log q)}$ and $\log 10/8 < 0.29$ and similarly for $5^\errs$.
\end{proof}

\begin{lemma}
  \label{lem:abstract_congruence_bound}
  Let $A, B \in \F[x]$ with $\gcd(A,B) = 1$, and $K_1 < K_2 < K_3 \in \ZZ_+$, as well as $q = |\F|$.
  Let $S$ denote the set of triples $(f_1,f_2,f_3) \in \F[x]^3$ such that
  $B f_2^2 \equiv f_1 f_3 \mod A$, while $\deg f_t \leq K_t$ and $f_2$ is monic.
  Then
  \[
  |S| \leq 2^{K_1+K_3}q^{K_2 + 1 + \max(0, \gamma)} \ ,
  \]
  where $\gamma = \max(K_1+K_3,\ 2K_2 + \deg B) - \deg A$.
\end{lemma}
\begin{proof}
  Consider first $\gamma < 0$ in which case $Bf_2^2 = f_1 f_3$.
  We can choose $f_2$ in $q^{K_2-1}$ ways.
  The prime divisors of $Bf_2^2$ should then be distributed among $f_1$ and $f_3$, which can be done in $2^{K_1 + K_3}$ ways.
  Finally, the leading coefficient of $f_1$ can be chosen in $q-1$ ways.

  Consider now $\gamma \geq 0$.
  We choose again first $f_2$ in one of $q^{K_2-1}$ ways.
  Then $f_1f_3$ must be in the set $\{ B f_2^2 + p A \mid p \in \F[x], \deg p \leq \gamma \}$, having cardinality at most $q^{\gamma+1}$.
  For each of these choices of $f_1 f_3$, we can again choose $f_1$ and $f_3$ in at most $(q-1)2^{K_1 + K_3}$ ways.
\end{proof}

The bound of \cref{prop:failure_prob_23} demonstrates a rapid, exponential decrease in the probability of failure as the number of errors decrease away from $\tau_\Pow(2,3)$.
The bound only becomes non-trivial a few errors below $\tau_\Pow(2,3)$, due to the term $0.29\errs/\log q - 1/4$ in the exponent.
For instance, for a $[64,27]$ code over $GF(64)$, then $\tau_\Pow(2,3) = 20\,\nicefrac{1}{4}$, but the bound is only less than $1$ for $\errs \leq 19$, also for the unsimplified bound \eqref{eqn:more_precise_bound}.
Such a penalty is not observed in simulations, however, and seems to be an artefact of our proof.
For the $[64,27]$ code, decoding succeeds almost always with 20 errors (see next section).
Similarly, for a $[256, 63]$ code over $GF(256)$, the bound is only non-trivial for $\errs < 108$, while \eqref{eqn:more_precise_bound} would be slightly better with $\errs < 110$; however, in simulations decoding works almost always up to $\floor{\tau_\Pow(2,3)} = 112$.

Nevertheless, in an asymptotic and relative sense, \cref{prop:failure_prob_23} guarantees that decoding up to $\tau_\Pow(2,3)$ almost always succeeds:

\begin{corollary}
  When $s=2$ and $\ell = 3$, with $n \rightarrow \infty$ while keeping $q/n$, $k/n$ and $\errs/n$ constant, the probability that Power decoding fails goes to 0 when $\errs/n < \tau_\Pow(2,3)/n$.
\end{corollary}
\begin{proof}
  Consider the high-error failure probability of \cref{prop:failure_prob_23}:
  \begin{IEEEeqnarray*}{rCl}
    4(q^{-8})^{(\tau_\Pow(2,3)-\errs) - (0.29\errs/\log q-1/4)}
    &\leq& 4(q^{-8n})^{\delta - 0.29\errs/(n\log q) + 1/(4n)} \ ,
  \end{IEEEeqnarray*}
  where $\delta = \tau_\Pow(2,3)/n - \errs/n$.
  Asymptotically $\delta$ approaches some positive constant, while the other terms in the exponent vanishes.
  The low-error case is similar.
\end{proof}

\subsection{Failure Behaviour in Relation to List Decoding}
  It is natural to ask if the failure behaviour of Power decoding is linked to whether or not there are multiple codewords close to the received word, i.e.~the list of codewords that e.g.~the Guruswami--Sudan algorithm would return.
  There seems, however, to be no clear relation like this, as we explain below.

  Consider that $\vec c \in \Code$ was sent and $\vec r$ was received.
  Suppose Power decoding has a decoding radius of $\tau$, and that we have a list decoder of the same decoding radius.
  Consider that $\vec c' \in \Code$ is another codeword and assume that all other codewords are farther from $\vec r$ than $\vec c$ or $\vec c'$.
  Then there are the following possibilities:
  \begin{enumerate}
    \item $\dist(\vec c, \vec r) = \dist(\vec c', \vec r) \leq \tau$ \label{itm:failure_equal}
    \item $\dist(\vec c', \vec r) < \dist(\vec c, \vec r) \leq \tau$ \label{itm:failure_beaten}
    \item $\dist(\vec c, \vec r) < \dist(\vec c', \vec r) \leq \tau$ \label{itm:failure_beating}
    \item $\dist(\vec c, \vec r) \leq \tau < \dist(\vec c', \vec r)$ \label{itm:failure_unique}
    \item $\dist(\vec c', \vec r) \leq \tau < \dist(\vec c, \vec r)$ \label{itm:failure_c_beyond}
    \item $\tau < \dist(\vec c, \vec r), \dist(\vec c', \vec r)$     \label{itm:failure_beyond}
  \end{enumerate}

  Clearly, both Power decoding and the list decoder will fail in recovering $\vec c$ in  \cref{itm:failure_c_beyond} and \cref{itm:failure_beyond}.
  In Items 1--4, the list decoder guarantees to recover $\vec c$ on a list, though for Items 1--3 that list will have length at least 2.

  For Power decoding it is less clear-cut.
  Firstly, for Items 1 and 2, then Power decoding ``fails'' according to the definition given at the beginning of \cref{sec:failure}.
  Indeed, for Item 2, then Power decoding guarantees to return $\vec c'$ or $\fail$.
  For Item 1, however, then \cref{alg:decoder} might be lucky and find $\vec c$, but in all likelihood the obtained solution to \cref{prob:ke_linearised} will be some linear combination of the two solutions corresponding to $\vec c$ and $\vec c'$; probably \cref{line:decoder:compute_f} or at least \cref{line:decoder:verify_dist} of \cref{alg:decoder} will return $\fail$.
  For Items 3 and 4, then Power decoding will probably obtain $\vec c$; but in either case, one can construct examples where it will fail.
  That is, whether or not there is only one codeword within radius $\tau$, then Power decoding might succeed or it might fail.

  That might be surprising at first, so we give examples of these cases.
  Consider $\Code$ to be the $[23,7]_{GF(23)}$ GRS code with evaluation points $\vec\alpha = (0,1,\ldots,22)$ and $\vec\beta = \vec 1$.
  For this code, $\tau_\Pow(2,3) = 9$.
  As an example for Item 3 where Power decoding succeeds, consider the following vectors:
  \begin{IEEEeqnarray*}{rCc;c;c;c;c;c;c;c;c;c;c;c;c;c;c;c;c;c;c;c;c;c;cc;l}
    \vec c_3 &=& (16& 15& 20& 20& 3& 0& 18& 0& 19& 16& 2& 11& 11& 3& 9& 18& 5& 0& 0& 0& 5& 0& 16&) & \in \Code \ , \\
    \vec r_3 &=& (16& 0& 20& 20& 0& 0& 18& 0& 19& 0& 2& 11& 0& 0& 0& 0& 5& 0& 0& 0& 5& 0& 0&) & \ .
  \end{IEEEeqnarray*}
  Then $8 = \dist(\vec r_3, \vec c_3) < \dist(\vec r_3,\vec 0) = \tau = 9$.
  So a list decoder with decoding radius 9 would obtain the list $[\vec 0, \vec c_3]$.
  Power decoding returns $\vec c_3$.
  Indeed, our observation is that Power decoding usually succeeds in the case of Item 3.

  An example for Item 4 where Power decoding fails, even though there is only one nearby codeword, $\vec 0$, is the following received word:
  \begin{IEEEeqnarray*}{rCc;c;c;c;c;c;c;c;c;c;c;c;c;c;c;c;c;c;c;c;c;c;cc;l}
    \vec r_{4} &=& (0& 2& 9& 1& 0& 0& 6& 0& 0& 0& 5& 0& 0& 0& 0& 0& 0& 4& 8& 15& 0& 0& 12) \ .
  \end{IEEEeqnarray*}
  This last example was found by random generation of error vectors of weight 9, after roughly $47\ 000$ successful decoding trials.
  As an aside, the failure probability bound of \cref{prop:failure_prob_23} gives the trivial bound 1 for the failure probability in this case.

\section{Simulation Results}
\label{sec:simulation}

\newcommand\GF[1]{GF(#1)}

The proposed decoding algorithm has been conceptually implemented in Sage v8.0 \cite{stein_sagemath_????}, and is available for download at \url{http://jsrn.dk/code-for-articles}.
The implementation follows the approach of \cref{sec:pade}, computing a solution basis using the Mulders--Storjohann algorithm \cite{mulders_lattice_2003}.
The asymptotic complexity of the implementation is therefore $O(\ell^3s^2n^2)$.

To evaluate the failure probability, we have selected a range of code and decoding parameters and run the algorithm for a large number of random errors.
More precisely, for each set of parameters, and each decoding radius $\tau$, we have created $N = 10^5$ random errors of weight exactly $\tau$ and attempted to decode a received word $\vec r = \vec c + \vec e$ for some randomly chosen $\vec c$ (though, of course, \cref{prop:power_gao_inv} implies that shifting by $\vec c$ makes no difference).
We have limited the decoding radii used to being $\floor{\tau_\Pow(s,\ell)} + \{-1,0,1\}$.
The results are listed as \cref{tbl:simulation}.

As is evident, $\tau_\Pow(s,\ell)$ very clearly describes the number of errors we can rely on correcting: the probability of failing appears to decay exponentially with $\tau_\Pow(s,\ell) - \errs$, as we might expect if extrapolating from the bound of \cref{prop:failure_prob_23}.
In fact, the failure probability is so low that it is difficult to observe failing cases for randomly selected errors.

The case having the highest failure rate is the very low-rate code $[21,3]_{\GF{23}}$.
For such a low-rate code, $\tau_\Pow(s,\ell)$ is quite close to the covering radius, and there is a significant probability that a random error will yield a received word which is closer to another codeword.
In this case, Power decoding always fails.
We performed another simulation for this code with $10^4$ random errors of weight exactly 14 and decoding using the Guruswami--Sudan list decoder.
This simulation gave a $16.1\%$ chance that another codeword was as close or closer to the sent codeword.
Thus most of the $19.7\%$ failures of Power decoding stem from this.

\begin{table}
  \centering
 \begin{tabular}{lllllll}
  $[n,k]_q$       & $(s,\ell)$ & $\tau_\Pow$             & $P_f(\floor{\tau_\Pow} - 1)$ & $P_f(\floor{\tau_\Pow})$ & $P_f(\floor{\tau_\Pow} + 1)$ & $\tau_{\rm bnd}$ \\ \midrule
$[21, 3]_{23}$    & $(6,19)$   & $14\,\nicefrac{1}{120}$ & $ 7.43 \times 10^{-3} $      & $ 1.97 \times 10^{-1} $  & $ 1 $                        &                  \\
$[24, 7]_{25}$    & $(2,3)$    & $10\,\nicefrac{1}{4}$   & $0$                          & $ 2.27 \times 10^{-3} $  & $ 1 $                        & 8 (9)            \\
$[32, 10]_{37}$   & $(2,4)$    & $13$                    & $0$                          & $ 2.78 \times 10^{-2} $  & $ 1 $                        &                  \\
$[64, 27]_{64}$   & $(2,3)$    & $20\,\nicefrac{1}{4}$   & $0$                          & $ 3.10 \times 10^{-4} $  & $ 1 $                        & 19 (19)              \\
$[68, 31]_{71}$   & $(3,4)$    & $20$                    & $0$                          & $0$                      & $ 1$                         &                  \\
$[125, 51]_{125}$ & $(4,6)$    & $42$                    & $0$                          & $0$                      & $ 1$                         &                  \\
$[256, 63]_{256}$ & $(2,4)$    & $116\,\nicefrac{2}{5}$  & $0$                          & $0$                      & $ 1- 3.00 \times 10^{-4} $   &
\end{tabular}
\caption{Simulation results. $P_f(\tau)$ denotes the observed probability of decoding failure (no result or wrong result) with random errors of weight exactly $\tau$.
  $\tau_{\rm bnd}$ indicates the number of errors $\errs$ for which \cref{prop:failure_prob_23} yields a bound $< 1$ (where applicable); in parentheses is if the probability estimate of \eqref{eqn:more_precise_bound} is used instead.
}
  \label{tbl:simulation}
\end{table}

\section{Efficient Solving of the Key Equations}
\label{sec:pade}

To solve \cref{prob:ke_linearised}, we will leverage existing algorithms by modelling \cref{prob:ke_linearised} as a \emph{simultaneous Hermite \Pade approximation} (SH \Pade), a well-studied computational problem.
This problem does not fit perfectly to \cref{prob:ke_linearised}, so to describe the modelling from one to the other we will introduce some technical notions pertaining to the solution sets of SH \Pade problems.
The upshot is \cref{alg:keyeq_of_pade} and \cref{cor:keyeq_of_pade} stating that we can rely completely on existing sophisticated algorithms to solve \cref{prob:ke_linearised} in complexity $\Oapp(\ell^\omega n)$ (or the faster $\Oapp(s^2 \ell^{\omega-1} n)$ if we rely on the unpublished \cite{rosenkilde_algorithms_2018}).

\begin{definition}[Simultaneous Hermite \Pade approximation]
  \label{def:pade}
  Given $A \in \F[x]^{s \times \ell}$ and $\Gamma_1,\ldots,\Gamma_{\ell} \in \F[x]$, as well as degree bounds $T_i, N_t \in \ZZ_{\geq 0}$, compute, if it exists, $\vec \lambda = (\lambda_1,\ldots,\lambda_s) \in \F[x]^s$ such that $\deg \lambda_i < T_i$, and
  \[
    \vec \lambda A \equiv \vec \psi \mod (\Gamma_1,\ldots,\Gamma_\ell) \ ,
  \]
  where $\vec \psi = (\psi_1,\ldots,\psi_\ell) \in \F[x]^\ell$ satisfies $\deg \psi_t < N_t$.
  (The modulo operation is element-wise, i.e.~the $i$'th entry of $\vec \lambda A$ is congruent to $\psi_i$ modulo $\Gamma_i$.)
\end{definition}

SH \Pade approximations have appeared elsewhere in coding theory: for the interpolation step of Guruswami--Sudan and the Wu decoding algorithms for Reed--Solomon and other codes \cite{zeh_interpolation_2011,chowdhury_faster_2015}, and for decoding of Hermitian algebraic-geometry codes \cite{nielsen_sub-quadratic_2015}.
Computing solutions to these very general forms of \Pade approximation goes back much further in the computer algebra community, though $\Gamma_i$ are usually powers of $x$, see e.g.~\cite{barel_general_1992,beckermann_uniform_1994} and the references therein.
In the generality above, the problem was first considered in \cite{nielsen_list_2013} solved using row reduction of $\F[x]$ matrices, and shortly thereafter in \cite{chowdhury_faster_2015} solved as an $\F$-linear system exhibiting block-Hankel structure \cite{bostan_solving_2008}.
Even more general notions include minimal approximant basis, or order basis \cite{giorgi_complexity_2003,gupta_triangular_2012}, and relation bases \cite{neiger_fast_2016}.

First we define a measure of how far a solution is from the degree bounds:
\begin{definition}
  For a given SH \Pade problem with $A \in \F[x]^{s \times \ell}$ as well as $\Gamma_i,T_i, N_t$, and a vector $\vec \lambda = (\lambda_1,\ldots,\lambda_s)$, the \emph{discrepancy} $\delta \in \ZZ$ of $\vec \lambda$ (wrt.~the SH \Pade problem) is given as:
  \[
    \delta = \max\Big( \max_{i=1,\ldots,s}(\deg \lambda_i - T_i),\ \max_{t=1,\ldots,\ell}(\deg \psi_t - N_t) \Big) \ ,
  \]
  where $\vec\psi = (\psi_1,\ldots,\psi_\ell) = \vec \lambda A \rem (\Gamma_1,\ldots,\Gamma_\ell)$.
\end{definition}
Note that a $\vec\lambda$ is a solution to an SH \Pade problem if and only if its discrepancy is negative.
We wish to link the type of degree restrictions of \cref{def:pade} with those of \cref{prob:ke_linearised}: for this, observe that a solution $\vec \lambda = (\lambda_1,\ldots,\lambda_s) \in \F[x]^s$ has a discrepancy of $\deg \lambda_1 - T_1$ if and only if
\begin{IEEEeqnarray*}{rCl;l}
  \deg \lambda_1 - T_1 &\geq& \deg \lambda_{i+1} - T_{i+1} & \textrm{ for } i = 1,\ldots,s-1 \\
  \deg \lambda_1 - T_1 &\geq& \deg \psi_t - N_t & \textrm{ for } t = 1,\ldots,\ell \ ,
\end{IEEEeqnarray*}
where $\vec \psi = (\psi_1,\ldots,\psi_\ell) = \vec\lambda A \rem (\Gamma_1,\ldots,\Gamma_\ell)$.
This leads to the following lemma:
\begin{lemma}
  \label{lem:keyeq_of_pade}
  Consider a received word $\vec r$, let $\tau \in \ZZ_{\geq 0}$ be a chosen decoding radius and assume at most $\tau$ errors occurred.
  Consider the SH \Pade approximation defined by $A = [ A_{i,t} ] \in \F[x]^{s \times \ell}$ as well as $\Gamma_t, T_i$ and $N_t$, given as:
  \begin{IEEEeqnarray*}{rCl}
    A_{i+1,t} &=&
    \left\{ \begin{array}{ll}
              \binom t i R^{t-i} G^i   & \textnormal{for } t = 1,\ldots,s-1 \textnormal{ and } i=0,\ldots, s-1\\
              \binom t i R^{t-i} G^i \modop G^s   & \textnormal{for } t = s,\ldots,\ell \textnormal{ and } i=0,\ldots, s-1 \\
            \end{array}
            \right. \\
  \Gamma_t &=&
    \left\{ \begin{array}{ll}
               x^{s \tau + t(n-1) + 1}  & \textnormal{for } t = 1,\ldots,s-1 \\
              G^s  & \textnormal{for } t = s, \ldots, \ell \\
            \end{array}
            \right. \\
  T_{i+1} &=& s \tau - i + 1 \quad\textnormal{ for } i = 0,\ldots,s-1 \\
  N_t &=& s\tau + t(k - 1) + 1 \quad\textnormal{ for } t = 1,\ldots,\ell \ .
  \end{IEEEeqnarray*}
  Let $\vec\lambda = (\lambda_1,\ldots,\lambda_s)$ be a solution to this SH \Pade approximation with minimal discrepancy among solutions whose discrepancy equals $\deg \lambda_1 - T_1$.
  Then $\vec \lambda, \vec \psi$ is a minimal solution to the instance of \cref{prob:ke_linearised} corresponding to $\vec r$, where $\vec\psi = \vec \lambda A \rem (\Gamma_1,\ldots,\Gamma_\ell)$.
\end{lemma}
\begin{proof}
  We first prove that $\vec\lambda$ is well-defined, that is there is such a solution to the SH \Pade problem.
  We claim that if $\vec{\lambda'}, \vec{\psi'}$ is a minimal solution to the instance of \cref{prob:ke_linearised} (which we know exists) then $\vec \lambda'$ is such a solution to the SH \Pade approximation.
  Note first by the assumption on decoding radius that $\deg \lambda_1' \leq s\tau$.
  Then Item 2 of \cref{prob:ke_linearised} implies $\deg \lambda'_{i+1} \leq s\tau - i < T_{i+1}$ for $i = 1,\ldots,s-1$.
  Similarly, by Item 3 then $\deg \psi'_t \leq s \tau + t(k-1) < N_t$ for $t = 1,\ldots,\ell$.
  Finally, observe that
  \begin{IEEEeqnarray}{rClCrCl}
    \deg \lambda'_{i+1} + i &\leq& \deg \lambda'_1   &\iff&  \deg \lambda'_{i+1} - T_{i+1} &\leq& \deg \lambda'_1 - T_1   \notag \\
    \deg \psi'_t - t(k-1) &\leq& \deg \lambda'_1 &\iff& \deg \psi'_t - N_t &\leq& \deg \lambda'_1 - T_1 \ .
    \label{eqn:keyeq_of_pade_lam}
  \end{IEEEeqnarray}
  so Item 2 and Item 3 of \cref{prob:ke_linearised} imply the discrepancy condition.

  The other direction is very similar: assume now $\vec\lambda$ is a solution to the SH \Pade problem with discrepancy $\deg\lambda_1-T_1$, and $\vec\psi = (\psi_1,\ldots,\psi_\ell) = \vec\lambda A \rem (\Gamma_1,\ldots,\Gamma_\ell)$.
  Item 1b of \cref{prob:ke_linearised} is obviously satisfied; for Item 1a, we know $\vec \lambda A_{*,t} \equiv \psi_t \mod \Gamma_t$ for $t = 1,\ldots,\ell$, where $A_{*,t}$ denotes the $t$'th column of $A$.
  Note that for $t < s$ we have
  \[
    \deg(\vec \lambda A_{*,t})
    \leq \max_{i=0,\ldots,s-1}\big((T_{i+1} - 1) + (\deg R^{t-i}G^i)\big)
    \leq s \tau + t(n-1) \ .
  \]
  For these values of $t$ we have $\Gamma_t = x^{s \tau + t(n-1) + 1}$ and so the congruence lifts to equality, i.e.~Item 1a.
  Item 2 and Item 3 of \cref{prob:ke_linearised}, follow directly from the discrepancy condition on $\vec\lambda$ together with \eqref{eqn:keyeq_of_pade_lam}.

  For minimality, assume conversely that $\deg\lambda'_1 < \deg \lambda_1$.
  Then $\deg \lambda'_1 - T_1 < \deg \lambda_1 - T_1$.
  Since $\vec \lambda'$ is a solution to the SH \Pade problem satisfying the discrepancy restriction, then this contradicts the minimality of $\vec\lambda$.
\end{proof}

\subsection{Solution bases for \Pade approximations}
\label{ssec:pade_solution}

\cref{lem:keyeq_of_pade} states that special solutions to a specific SH \Pade problem are actually minimal solutions to \cref{prob:ke_linearised}.
Many algorithms for solving SH \Pade problems, and in particular the fastest ones known, actually find a basis of \emph{all} solutions, for a notion of ``basis'' which we introduce momentarily.
We will now show that such a basis must contain a solution satisfying the constraints of \cref{lem:keyeq_of_pade} and hence will be a minimal solution to \cref{prob:ke_linearised}.
This section uses a number of concepts which are standard in polynomial matrix literature, but less so in coding theory.
They will not be used outside this section.

The degree of a vector $\vec v \in \F[x]^m$ or matrix $A \in \F[x]^{m' \times m}$ is the maximal degree of its entries.
The \emph{leading matrix} of $A$, denoted $\LM(A) \in \F^{m' \times m}$, has $(i,j)$'the entry equal to the coefficient of $x^{d_i}$ of $A_{i,j}$, where $d_i$ is the degree of the $i$'th row of $A$.
The \emph{leading indices} of $\vec v$, denoted $\pivs(\vec v) \subset \{ 1,\ldots,m \}$, are the indices of $\vec v$ which have degree $\deg \vec v$.
In other words, $\LM(A)$ is non-zero exactly at the leading indices of the rows of $A$.
We also introduce \emph{shifted} variants of the above notions: given a ``shift'' $\vec h \in \ZZ^m$, then $\deg_{\vec h} \vec v := \deg(\vec v x^{\vec h})$, where $x^{\vec h}$ is the diagonal matrix with entries $x^{h_1},\ldots,x^{h_m}$.
Similarly $\deg_{\vec h} A := \deg(A x^{\vec h})$; $\LM_{\vec h}(A) := \LM(A x^{\vec h})$; and $\pivs_{\vec h}(\vec v) := \pivs(\vec v x^{\vec h})$.
Note that if $\vec h$ has negative entries, this notation may formally pass over the ring of Laurent series.

To relate some of these concepts to the previous section, note that if we set $\vec h = (-T_1,\ldots,-T_s,-N_1,\ldots,-N_\ell)$, then for any $\vec \lambda \in \F[x]^s$ the discrepancy of $\vec\lambda$ is exactly $\deg_{\vec h}(\vec \lambda | \vec \psi)$, where $\vec \psi = \vec \lambda A \rem (\Gamma_1,\ldots,\Gamma_\ell)$.
Further, the requirement of \cref{lem:keyeq_of_pade}, i.e.~that the discrepancy of $\vec\lambda$ equals $\deg \lambda_1 - T_1$, is equivalent to requiring $1 \in \pivs_{\vec h}(\vec \lambda | \vec \psi)$.

\begin{lemma}
  \label{lem:sh_pade_solutions}
  Consider a given SH \Pade problem with $A \in \F[x]^{s \times \ell}$ as well as $\Gamma_t, T_i$ and $N_t$.
  A vector $\vec \lambda \in \F[x]^s$ is a solution iff there is $\vec \psi \in \F[x]^\ell$ such that $\deg_{\vec h}(\vec\lambda | \vec \psi) < 0$ and $(\vec \lambda | \vec \psi)$ is in the row space of $M$, where
  \begin{IEEEeqnarray*}{rCl}
    M &=& \left[
        \begin{array}{@{}c|c@{}}
          I_{s \times s} & A \\\hline
          & \diag(\Gamma_1,\ldots,\Gamma_\ell) ,
        \end{array}\right] \\
    \vec h &=& (-T_1,\ldots,-T_s,-N_1,\ldots,-N_\ell) \ .
  \end{IEEEeqnarray*}
\end{lemma}
\begin{proof}
  The row space of $M$ is exactly the vectors $(\vec\lambda | \vec \psi)$ where $\vec \lambda A \equiv \vec \psi \mod (\Gamma_1,\ldots,\Gamma_\ell)$.
  For such vectors, then $\vec \lambda$ is solution when it has negative discrepancy, which by our earlier observation is exactly when it has negative $\vec h$-degree.
\end{proof}

We say that matrix $A$ is \emph{$\vec h$-row reduced} if $\LM_{\vec h}(A)$ has full row rank, see \cite[Ch.~6.3.2]{kailath_linear_1980} and \cite{barel_general_1992}.
For any $M \in \F[x]^{m' \times m}$ with full row rank, there always exists another matrix $A$ which is $\vec h$-row reduced and has the same row space as $M$; see e.g.~\cite{mulders_lattice_2003} for a succinct iterative algorithm.
Row reduced matrices derive their interest from having minimal row degrees of all possible bases of the same row space; this property can be generalised to the \emph{predictable degree property} \cite[Ch.~6.3.2]{kailath_linear_1980}, of which we will use the following variant:

\begin{proposition}
  \label{prop:pdp}
  Let $\vec h \in \ZZ^m$ be a shift, let $A \in \F[x]^{m' \times m}$ be row reduced, and let $\vec a_1,\ldots,\vec a_{m'}$ be the rows of $A$.
  Let $\vec v \in \F[x]^m$ be any vector in the row space of $A$.
  Then there exists $\vec q = (q_1,\ldots,q_{m'}) \in \F[x]^{m'}$ such that $\vec v = \vec q A$ and
  \begin{IEEEeqnarray*}{rCl}
    \deg_{\vec h} \vec v &=& t \ , \textrm{ where } t = \max_{i=1,\ldots,m'}(\deg q_i + \deg_{\vec h} \vec a_i) , \ \textrm{ and } \\
    \pivs_{\vec h}(\vec v) &\subseteq& \bigcup_{i \in I} \pivs_{\vec h}(\vec a_i) \ , \textrm{ where }
    I = \{ i \mid \deg q_i = \deg_{\vec h} \vec v - \deg_{\vec h} \vec a_i \} \ .
  \end{IEEEeqnarray*}
\end{proposition}
\begin{proof}
  The existence of $\vec q$ is trivial since $\vec v$ is in the row space of $A$.
  Note that
  \begin{equation}
    \label{eqn:pdp}
    \vec v x^{\vec h} = \vec q A x^{\vec h} \ ,
  \end{equation}
  and so $\deg_{\vec h}\vec v \leq t$.
  Let $\vec{\grave q} \in \F^m$ be the scalar vector whose $i$'th entry is the leading coefficient of $q_i$ if $i \in I$ and 0 otherwise.
  Let $\vec{\grave v} \in \F^m$ be the scalar vector of $x^t$'th coefficients of $\vec v x^{\vec h}$.
  Then \eqref{eqn:pdp} implies $\vec{\grave v} = \vec{\grave q}\LM_{\vec h}(A)$.
  But then $\vec{\grave v} \neq \vec 0$ since $\LM_{\vec h}(A)$ has full row rank.
  Therefore $\deg_{\vec h}(\vec v) = t$.
  Further, $\pivs_{\vec h}(\vec v)$ is then the indices of non-zero entries of $\vec{\grave v}$, i.e.~the non-zero entries of $\vec{\grave q}\LM_{\vec h}(A)$, i.e.~a subset of the non-zero columns of $\LM_{\vec h}(A')$, where $A'$ is the rows of $A$ indexed by $I$.
\end{proof}

We are now in a position to define a notion of ``basis'' of all solutions:
\begin{definition}
  Consider a given SH \Pade problem with $A \in \F[x]^{s \times \ell}$ as well as $\Gamma_t, T_i$ and $N_t$.
  Let $B'$ be any matrix which is left-equivalent\footnote{%
    I.e.~there exists an invertible matrix $U \in \F[x]^{(s+\ell)\times(s+\ell)}$ such that $B' = UM$.
  } to $M$ and $\vec h$-row reduced, where $M$ and $\vec h$ are as in \cref{lem:sh_pade_solutions}.
  Let $B \in \F[x]^{m \times (s+\ell)}$ consist of the rows of $B'$ with negative $\vec h$-degree.
  Then $B$ is a \emph{solution basis} to the SH \Pade problem.
\end{definition}

Not only will the rows of a solution basis $B$ be solutions themselves; the main point is that they will \emph{span} every single solution in a predictable way:
any solution must be a linear combination of the rows of the complete, $\vec h$-row reduced matrix $B'$, but due to the predictable degree property, for the $\vec h$-degree of a vector to be negative, it must be spanned only by vectors with negative $\vec h$-degree themselves, and with coefficients of bounded degree.

We now see an easy algorithm for solving SH \Pade problems: set up $M$ and compute a row reduced matrix $B'$ which is left-equivalent to $M$.
This could e.g.~be done using the iterative Mulders--Storjohann algorithm \cite{storjohann_high-order_2003}, or using the reduction from row reduction to order basis \cite{gupta_triangular_2012,giorgi_complexity_2003}.
The latter yields a complexity of $\Oapp( (s+\ell)^\omega D )$, where $D = \max T_i + \max \deg g_t$.

We will continue the discussion a bit further, since a result from \cite{rosenkilde_algorithms_2018} -- which is not yet published -- allows a faster algorithm if we only compute the first $s$ columns of a solution basis:

\begin{definition}
  \label{def:solution_spec}
  Consider a given SH \Pade problem with $A \in \F[x]^{s \times \ell}$ as well as $\Gamma_t, T_i$ and $N_t$.
  A \emph{solution specification} is a matrix $L \in \F[x]^{m \times s}$ and discrepancies $\delta_1,\ldots,\delta_m < 0$ such that there is a matrix $\grave B \in \F[x]^{m \times \ell}$ for which $[ L \mid \grave B ]$ is a solution basis, and whose rows have $\vec h$-degree $\delta_1,\ldots,\delta_m$.
\end{definition}

\begin{proposition}[\!\!\protect{\cite{rosenkilde_algorithms_2018}}]
  \label{prop:pade_complexity}
  Consider an SH \Pade approximation problem with $A \in \F[x]^{s \times \ell}$ as well as $\Gamma_i, T_i, N_t$, satisfying $s < \ell$, and $T_i < \deg \lcm(\Gamma_1,\ldots,\Gamma_\ell)$ for $i=1,\ldots,s$, and $\deg N_t < \deg \Gamma_t$ for $t=1,\ldots,\ell$.
  There exists an algorithm which computes a solution specification using
  \[
    O(\ell^{\omega-1} \PolyMult{sD} (\log sD)(\log sD/\ell)^2) \subset \Oapp(s\ell^{\omega-1}D)
  \]
  operations in $\F$, where $D = \max_i T_i + \max_t \deg \Gamma_t$.
\end{proposition}

To solve \cref{prob:ke_linearised} using SH \Pade approximations in the complexity of \cref{prop:pade_complexity}, the only remaining piece is to prove that a solution specification must contain a row for which we can apply \cref{lem:keyeq_of_pade}.

\begin{proposition}
  Consider an SH \Pade approximation problem with $A \in \F[x]^{s \times \ell}$ as well as $\Gamma_i, T_i, N_t$.
  If there exists a solution $\vec\lambda \in \F[x]^s$ such that its discrepancy equals $\deg\lambda_1 - T_1$, then such a solution with minimal discrepancy will appear in a solution specification.
\end{proposition}
\begin{proof}
  Let $B \in \F[x]^{(s+\ell) \times (s+\ell)}$ be a completed, $\vec h$-row reduced matrix, left-equivalent to $M$, corresponding to the solution specification, where $\vec h$ and $M$ are as in \cref{lem:sh_pade_solutions}.
  Let $\vec \lambda \in \F[x]^s$ be a solution with minimal discrepancy satisfying $\deg\lambda_1 - T_1$.
  Then there is $\vec \psi \in \F[x]^\ell$ and $\vec q = (q_1,\ldots,q_{s+\ell}) \in \F[x]^{s + \ell}$ such that $(\vec \lambda | \vec\psi) = \vec q B$ and $\deg_{\vec h}(\vec \lambda | \vec \psi) = \deg\lambda_1 - T_1$, i.e.~$1 \in \pivs_{\vec h}(\vec \lambda | \vec \psi)$.
  By \cref{prop:pdp} then there is a row $\vec b_i$ of $B$ with $1 \in \pivs_{\vec h}(\vec \lambda | \vec \psi)$ and $q_i \neq 0$.
  Then $\deg_{\vec h} \vec b_i \leq \deg_{\vec h}(\vec \lambda | \vec \psi)$, so if we write $\vec b_i = (\vec \lambda' | \vec \psi')$ then $\vec \lambda'$ is a solution with discrepancy $\deg \lambda'_1 - T_1 \leq \deg \lambda_1 - T_1$.
  To not contradict our choice of $\vec\lambda$, then equality must hold, and $\vec \lambda'$ is a satisfactory solution appearing in a solution specification.
\end{proof}

A complete algorithm for solving Problem 4 using solution specifications of SH \Pade approximations is given as \cref{alg:keyeq_of_pade}.

\begin{algorithm}[t]
  \caption{Solving \cref{prob:ke_linearised} using SH \Pade approximation}
  \label{alg:keyeq_of_pade}
  \begin{algorithmic}[1]
    \Require{$R, G \in \F[x]$, $s, \ell, \tau \in \ZZ_+$ with $s \leq \ell$.}
    \Ensure{A minimal solution $(\vec \lambda | \vec \psi)$ to \cref{prob:ke_linearised} if one exist, or $\fail$\;}
    \State Compute $A \in \F[x]^{s \times \ell}$ as in \cref{lem:keyeq_of_pade}, and set $\Gamma_t, T_{i+1}, N_t$ as in that lemma.
    \State $L, \vec \delta \ass $ solution specification to the SH \Pade problem of $A$ and $\Gamma_t, T_{i+1}, N_t$. \label{line:alg_keyeq_of_pade}
    \State $\vec\lambda \ass $ a minimal $\vec h$-degree row of $L$ among rows with $1 \in \pivs_{\vec h}(\vec \lambda)$, where $\vec h$ is as in \cref{lem:sh_pade_solutions}.
    If there is no such row, \Return $\fail$.
    \State $\vec \psi \ass \vec \lambda A \rem (\Gamma_1,\ldots,\Gamma_\ell)$.
    \State \Return $(\vec\lambda | \vec \psi)$.
  \end{algorithmic}
\end{algorithm}

\begin{corollary}
  \label{cor:keyeq_of_pade}
  \cref{alg:keyeq_of_pade} is correct.
  It has complexity $O(\ell^\omega \PolyMult{s n} (\log sn)) \subset \Oapp(s\ell^{\omega}n)$ operations in $\F$, using \cite{gupta_triangular_2012,giorgi_complexity_2003} for \cref{line:alg_keyeq_of_pade}.
  Using the results of \cite{rosenkilde_ne_nielsen_algorithms_2016}, it has complexity $O(\ell^{\omega-1} \PolyMult{s^2n} (\log sn)(\log sn/\ell)^2) \subset \Oapp(s^2\ell^{\omega-1}n)$.

  The same expressions hold for the complexity of \cref{alg:decoder}.
\end{corollary}
\begin{proof}
  Correctness follows from the results of this section and the last;
  note that the requirements of \cref{prop:pade_complexity} are satisfied in our case and that $D \in O(sn)$.
  For complexity, we merely need to argue that computing the solution specification of the SH \Pade approximation dominates.
  Indeed, $A$ can be computed using dynamic programming in $O(\ell s \PolyMult{s n})$ by remarking that $R^{t-i} G^i$ can be computed as the product of two previously computed terms of roughly half the size.
  The only other non-trivial computation is that of $\vec\psi$ which can also be carried out in complexity $O(\ell s \PolyMult{s n})$.
\end{proof}

\begin{remark}
  For short block-lengths, it can be of interest to consider the computational complexity when not using fast arithmetic, i.e.~taking $\PolyMult{n} = O(n^2)$ and $n^\omega = O(n^3)$.
  In this regime, a much simpler algorithm than those mentioned in \cref{cor:keyeq_of_pade} is to compute a solution basis by applying the Mulders--Storjohann row-reduction algorithm \cite{mulders_lattice_2003}.
  This yields the complexity $O(\ell^3 s^2 n^2)$, which is similar to complexities of interpolation algorithms for the Guruswami--Sudan in this regime, see e.g.~\cite{nielsen_decoding_1998,zeh_interpolation_2011,lee_list_2008}.
\end{remark}

\section{Re-Encoding}
\label{sec:reencoding}

\def\hc{\hat{\vec c}}
\def\hr{\hat{\vec r}}
\def\hf{\hat f}

``Re-Encoding'' is a simple technique invented by K\"otter and Vardy, originally for reducing the computational burden of the interpolation step in the Guruswami--Sudan algorithm \cite{kotter_complexity_2003}.
It is especially powerful when using different multiplicities at each point, such as in the K\"otter--Vardy soft-decision decoding version of Guruswami--Sudan \cite{kotter_algebraic_2003}.
For the regular Guruswami--Sudan, and in usual asymptotic analysis where $k/n$ is considered a constant, re-encoding does not change the asymptotic cost; however, it can have a significant practical impact on the running time, especially for higher-rate codes.
We will now show that the re-encoding transformation easily applies to Power decoding as well.

Consider that $\hr$ is the received word.
Using Lagrange interpolation, we can easily compute the unique $\hc = \ev(\hf) \in \Code$ such that $\hc$ and $\hr$ coincide on the first $k$ positions.
Clearly, decoding $\vec r = \hr - \hc$ immediately gives a decoding of $\hr$, and thanks to \cref{prop:power_gao_inv} we know Power decoding will succeed on $\vec r$ if and only if it succeeds on $\vec \hr$.
The idea of re-encoding is that the leading $k$ zeroes of the resulting $\vec r$ might be utilised in the decoding procedure to reduce the computation cost of decoding $\vec r$.

Assume therefore for this section that $\vec r$ is the received word after re-encoding and therefore has $k$ leading zeroes.
That means $\hat G \mid R$ where $\hat G = \prod_{i=1}^k (x-\alpha_i)$.
Consider the linearised key equations of \cref{prob:ke_linearised}.
Each of them are now divisible by $\hat G^{\min(s,t)}$, and so we obtain:

\begin{IEEEeqnarray*}{r?l?l}
  1a') & \psi_t/\hat G^t         =      \sum_{i=0}^{t} \lambda_{i+1} \cdot \left(\binom t i R^{t-i} G^i/\hat G^t\right)  \ ,           & \textrm{for } t=1, \ldots, s-1  \\
  1b') & \psi_t/\hat G^s         \equiv \sum_{i=0}^{s-1} \lambda_{i+1} \cdot \left(\binom t i R^{t-i} G^i/\hat G^s\right) \mod (G/\hat G)^s \ , & \textrm{for } t=s, \ldots, \ell \ .
\end{IEEEeqnarray*}

The elements $\grave \psi_t \defeq \psi_t/\hat G^{\min(s,t)}$ and $R^{t-i} G^i/\hat G^{\min(s,t)}$ are all polynomials, but of much lower degree than before.
Thus, we can solve for $\lambda_i$ and $\grave \psi_t$ directly which has fewer coefficients.
The degree restriction on $\grave \psi_t$ becomes
\[
  \deg \lambda_1 + t(k-1) - \min(s,t)k \geq \deg \grave \psi_t \ .
\]

The complete decoding algorithm is exactly as \cref{alg:decoder} with \cref{line:decoder:solve_keyeq} replaced by the re-encoded key equations, and where $f$ in \cref{line:decoder:compute_f} can be computed as $ f = \grave \psi_1 \hat G / \lambda_1$.

To solve the re-encoded key equations, we proceed exactly as before: the following is an analogue of \cref{lem:keyeq_of_pade} linking the restrictions on $\lambda_i$ and $\grave \psi_t$ to an SH \Pade approximation, whose proof is analogous to that of  \cref{lem:keyeq_of_pade}.

\begin{lemma}
  \label{lem:reenc_of_pade}
  Consider a received word $\vec r$ whose first $k$ positions are 0, let $\tau \in \ZZ_{\geq 0}$ be a chosen decoding radius and assume at most $\tau$ errors occurred.
  Consider the SH \Pade approximation defined by $A = [ A_{i,t} ] \in \F[x]^{s \times \ell}$ as well as $\Gamma_t, T_i$ and $N_t$, given as:
  \begin{IEEEeqnarray*}{rCl}
    A_{i+1,t} &=&
    \left\{ \begin{array}{ll}
              \binom t i R^{t-i} G^i/\hat G^{t}   & \textnormal{for } t = 1,\ldots,s-1 \textnormal{ and } i=0,\ldots, s-1\\
              \binom t i R^{t-i} G^i/\hat G^{s} \modop (G/\hat G)^s   & \textnormal{for } t = s,\ldots,\ell \textnormal{ and } i=0,\ldots, s-1 \\
            \end{array}
            \right. \\
  \Gamma_t &=&
    \left\{ \begin{array}{ll}
               x^{s \tau + t(n-k-1) + 1}  & \textnormal{for } t = 1,\ldots,s-1 \\
              (G/\hat G)^s  & \textnormal{for } t = s, \ldots, \ell \\
            \end{array}
            \right. \\
  T_{i+1} &=& s \tau - i + 1 \quad\textnormal{ for } i = 0,\ldots,s-1 \\
  N_t &=& s\tau + t(k - 1) - \min(s,t)k + 1 \quad\textnormal{ for } t = 1,\ldots,\ell \ .
  \end{IEEEeqnarray*}
  Let $\vec\lambda = (\lambda_1,\ldots,\lambda_s)$ be a solution to this SH \Pade approximation with minimal discrepancy among solutions whose discrepancy equals $\deg \lambda_1 - T_1$.
  Then $\vec \lambda, \vec{\grave \psi}$ is a minimal solution to the instance of \cref{prob:ke_linearised} corresponding to $\vec r$, where $\vec{\grave \psi} = \vec \lambda A \rem (\Gamma_1,\ldots,\Gamma_\ell)$.
\end{lemma}

As mentioned, the asymptotic complexity of solving the SH \Pade approximation of \cref{lem:reenc_of_pade} is not lower than that of \cref{cor:keyeq_of_pade}, in the usual asymptotic regime where we take $k \in \Theta(n)$.
However, considering \cref{prop:pade_complexity}, the expression $D$ becomes $s(\tau + n-k)$ rather than $s(\tau + n)$.
This should give a noticeable, constant factor speed-up for the complete decoding algorithm.

\ifExtended

\section{Syndrome Key Equations}
\label{sec:syndrome}

As described in \cref{sec:preliminaries}, the first key equation decoding algorithm was based on the notion of syndrome polynomial \cite{berlekamp_algebraic_1968}, and similarly, Power decoding without multiplicities was first described using a similar list of key equations \cite{schmidt_syndrome_2010}.
The key equations of \cref{thm:power_key} can similarly be rewritten to be based on syndrome polynomials, which we will show in this section.
As is usual for syndrome-formulated key equations, we will assume that $0$ is not used as an evaluation point.
Therefore $x \nmid G$.
Furthermore, due to a non-essential technicality, we will assume $s < n$.
If this did not hold, the following analysis of parameters would be slightly more complicated but not impossible.

Recall the reversal operator $\rev[d]{p}$ which we defined in \cref{ssec:classical_ke}.
Define for a given value of the multiplicity $s$ the following variants of the powered Lagrange interpolant $R$ as well as a generalised notion of syndrome:

\begin{IEEEeqnarray*}{rCl+rCl}
  \label{eqn:syndrome_generalised}
  R\T{i,t} &\defeq&   R^{t-i} \modop G^{s-i} &
  S\T{i,t} &=& \frac {\rev{R\T{i,t}}} {\rev G^{s-i}} \ .
\end{IEEEeqnarray*}
Note the degree that the reversal-operator on $\rev{R\T{i,t}}$ uses: if $t-i \leq s-i$ then $R\T{i,t} = R^{t-i}$ so the degree upper bound is $(t-i)(n-1)$. 
If $t-i > s-i$ then $\deg R^{t-i} > \deg G^{s-i}$ since we have assumed $s < n$, and therefore $\deg R\T{i,t} \leq (s-i)n - 1$.

If $s=1$ then $S\T{1,1}$ equals the classical syndrome polynomial $S$ which we used in \cref{ssec:classical_ke}, and $S\T{1,t}$ equals the syndromes $S\T{t}$ discussed in \cref{ssec:power_ke}.
The syndromes $S\T{i,t}$ also appear (with a slightly different definition) in the Interpolation key equations for Guruswami--Sudan by Gentner et al.~\cite{zeh_interpolation_2011}.
We can then formulate the---markedly more involved---syndrome variant of \cref{thm:power_key}:

\begin{theorem}
  \label{thm:power_key_syn}
  For any $s, \ell \in \ZZ_+$ with $\ell \geq s$, then there exist $g_t \in \F[x]$ for $t=s,\ldots,\ell$ such that
  \begin{IEEEeqnarray*}{lCl+l}
    \sum_{i=0}^t \rev{\Lambda^{s-i}\Omega^i} \left(\binom t i S\T{i,t} \right) &\equiv& 0 \mod x^{\varrho_t}
                  & \textrm{for } t=1, \ldots, s-1  \ , \\
    \sum_{i=0}^t \rev{\Lambda^{s-i}\Omega^i} \left(\binom t i S\T{i,t} x^{\iota_{i,t}} \right)  &\equiv& g_t \mod x^{\varrho_t}
                  & \textrm{for } t=s, \ldots, \ell \ ,
  \end{IEEEeqnarray*}
  where
  \begin{IEEEeqnarray*}{rCl}
    \deg g_t &\leq& \left\{ \begin{array}{ll}
                       \errs s - s     & \textrm{if } t = s \\
                       \errs s - 1     & \textrm{if } t > s
                       \end{array} \right.
                     \\
    \varrho_t &=& \left\{ \begin{array}{ll}
                       t(n-k)          & \textrm{if } t \leq s \\
                       sn - t(k-1) - 1 & \textrm{otherwise}
                       \end{array} \right.
                     \\
    \iota_{i,t} &=& \left\{ \begin{array}{ll}
                          0 & \textrm{if } t = s \\
                          i & \textrm{if } t > s
                       \end{array} \right. .
  \end{IEEEeqnarray*}
\end{theorem}
\begin{proof}
  We need to distinguish between two cases: $t < s$ and $t \geq s$.
  Assume first $t<s$.
  Since $R\T{i,t} = R^{t-i}$, \cref{thm:power_key} gives us
  \begin{IEEEeqnarray*}{rCl+l}
    \textstyle \sum_{i=0}^t \big(\Lambda^{s-i}\Omega^i\big)\big(\tbinom t i R\T{i,t} G^i\big)  &=& \Lambda^s f^t  & \iff \\[3pt]
    \textstyle \rev[\errs s + t(n-1)]{\sum_{i=0}^t \big(\Lambda^{s-i}\Omega^i\big)\big(\tbinom t i R\T{i,t} G^i\big)} &=& \rev[\errs s + t(n-1)]{\Lambda^s f^t} \ ,
  \end{IEEEeqnarray*}
  where $\errs s + t(n-1)$ arise from counting the degree upper bound on the left-hand side.
  Every term in the sum has the same degree bound, so we get
  \begin{IEEEeqnarray*}{rCl+l}
    \textstyle \sum_{i=0}^t \rev{\Lambda^{s-i}\Omega^i}\big(\tbinom t i \rev{R\T{i,t}} \rev G^i\big) &=& \rev{\Lambda^s f^t} x^{t(n-k)} &\implies \\
    \textstyle \sum_{i=0}^t \rev{\Lambda^{s-i}\Omega^i}\big(\tbinom t i \rev{R\T{i,t}} \rev G^i\big) &\equiv& 0 \mod x^{t(n-k)} &\iff \\
    \textstyle \sum_{i=0}^t \rev{\Lambda^{s-i}\Omega^i}\big(\tbinom t i S\T{i,t} \big) &\equiv& 0 \mod x^{t(n-k)} \ ,
  \end{IEEEeqnarray*}
  where the last line follows from $\rev G^s$ being invertible modulo $x^{t(n-k)}$.
  This concludes the case $t < s$.

  For the case $t \geq s$, we proceed similarly.
  In the congruence of \cref{thm:power_key}, we can readily replace $R^{t-i}G^i$ with $R\T{i,t}G^i$ modulo $G^s$.
  This gives us:
  \begin{IEEEeqnarray*}{rCl+l}
    \Lambda^s f^t & \equiv & \sum_{i=0}^{s-1} \big(\Lambda^{s-i}\Omega^i\big)\left(\binom t i R\T{i,t} G^i\right)  \mod G^s  &\implies \\
    \Lambda^s f^t + \rev{g_t} G^s & = & \sum_{i=0}^{s-1} \big(\Lambda^{s-i}\Omega^i\big)\left(\binom t i R\T{i,t} G^i\right) \ ,
  \end{IEEEeqnarray*}
  for some $\rev{g_t} \in \F[x]$.
  The degree of the right-hand side is bounded as
  \[
  \max_i\big\{ (s-i)\errs + i(\errs - 1) + \deg R\T{i,t} + in \big\}
     \leq \left\{
       \begin{array}{ll}
         s\errs + s(n-1) & \textrm{if } t=s \\
         s\errs + sn-1   & \textrm{if } t>s
       \end{array}
       \right.
       \ .
  \]
  This immediately bounds $\deg g_t$ as the theorem states.
  Note that the above equals $\varrho_t + s\errs + t(k-1)$ in all cases.
  We can now reverse the equation as in the previous case.
  When $t > s$ then the degree bound on the summands are not all the same, so we must add powers of $x$ in the reversed expression:
  \begin{IEEEeqnarray*}{rCl+l}
    \rev{\Lambda^s f^t} x^{\varrho_t} + g_t \rev G^s &=& \sum_{i=0}^{s-1} \rev{\Lambda^{s-i}\Omega^i} \left(\binom t i \rev{R\T{i,t}} \rev G^i x^{\iota_{i,t}} \right)  &\implies \\
    g_t \rev G^s & \equiv & \sum_{i=0}^{s-1} \rev{\Lambda^{s-i}\Omega^i} \left(\binom t i \rev{R\T{i,t}} \rev G^i x^{\iota_{i,t}} \right)  \mod x^{\varrho_t} &\iff \\
    g_t & \equiv & \sum_{i=0}^{s-1} \rev{\Lambda^{s-i}\Omega^i} \left(\binom t i S\T{i,t} x^{\iota_{i,t}} \right)  \mod x^{\varrho_t}
     \ .
  \end{IEEEeqnarray*}
\end{proof}

\begin{remark}
  Just as we remarked that the key equations of \cref{thm:power_key} resemble certain characterisations of Interpolation polynomials in Guruswami--Sudan, so does the above syndrome formulation resemble the syndrome Interpolation key equations of \cite{zeh_interpolation_2011}.
  Again, the deeper relation between the error locator approach and the Guruswami--Sudan is unclear.
\end{remark}

\cref{thm:power_key_syn} leads to a decoding algorithm in the very same way as \cref{thm:power_key}.
We could call these algorithms ``Power syndromes'' and ``Power Gao'' respectively.
We have the following important remark:

\begin{corollary}
  Decoding using Power Gao succeeds if and only if decoding using Power syndromes succeeds.
\end{corollary}
\begin{proof}
  This follows easily by the same transformation as in the proof of \cref{thm:power_key_syn}: a solution to the linearised key equations of Power Gao induces a solution to the linearised key equations of Power syndromes, and vice versa.
\end{proof}

Thus the two decoding algorithms have exactly the same decoding performance.

We won't go through the details of a syndrome-decoding analogue of \cref{lem:keyeq_of_pade}, but it follows along exactly the lines we have seen in both \cref{sec:pade} and \cref{sec:reencoding}.
The asymptotic complexity is again the same as that of \cref{cor:keyeq_of_pade}, but one would again expect a noticeable, constant factor speed-up very similar to that of the re-encoding transformation.

\fi

\section{Conclusion}

We demonstrated how the Power decoding technique for Reed--Solomon codes can be augmented with a new parameter---the multiplicity---to attain the Johnson decoding radius \cite{johnson_new_1962}.
The resulting decoder is, as the original Power decoding algorithm of Schmidt, Sidorenko and Bossert \cite{schmidt_syndrome_2010}, a partial decoder which fails for a few error patterns within its decoding radius.
We showed how one can efficiently solve the resulting key equations using existing algorithms for simultaneous Hermite \Pade approximation problems.

The proposed decoder has applications as a simpler alternative to the Guruswami--Sudan algorithm---especially for hardware implementations---due to its simple one-step shift-register type structure.
In particular, for medium rate codes, Power decoding with a low multiplicity of 2 or 3 would not be much more complicated to implement than half-the-minimum distance decoding while still offering a significant improvement in decoding performance.
For such parameters, the root-finding step of the Guruswami--Sudan algorithm would have a significant circuit area and latency.

The exact failure behaviour of the decoding method remains largely open.
For $s=1$, i.e.~the original Power decoding, the failure probability has previously been bounded only for $\ell=2,3$.
The case $s>1$ seems no easier to analyse:
\cref{prop:power_gao_inv} simplifies the equations one needs to analyse, and this was instrumental in the case for which we were able to bound the failure probability: $(s,\ell) = (2,3)$.
For these parameters, the decoding radius improves upon the case $s=1$ whenever the rate is within $]1/6 ; 1/2[$.
The claimed decoding radius of the decoder for other parameters was backed by simulations on a range of codes: this demonstrates a failure probability which seems to decay exponentially as the number of errors is reduced.

\ifExtended
We also discussed two variants of the decoding method which reduces the cost in practice: re-encoding and a syndrome formulation.
Either method roughly replaces the complexity dependency on $n$ with $n-k$.
More detailed analysis, and concrete choices of basis reduction algorithms is necessary to determine which one is fastest in practice.
\else
We also discussed how re-encoding can directly be applied to bring a performance improvement in practice: roughly, it will replace the $n$ by $n-k$ in the complexity expressions.
\fi

The proposed decoding algorithm has already been adapted to improve decoding of Interleaved RS codes \cite{puchinger_decoding_2017}.
Power decoding has previously been applied to other codes as well, e.g.~Complex RS codes \cite{mohamed_deterministic_2015}, and it seems clear that the proposed addition of multiplicities can aid those applications as well.
Another interesting question is to extend Power decoding to soft-decision decoding, similar to K\"otter--Vardy's variant of  the Guruswami--Sudan algorithm \cite{kotter_algebraic_2003}.

\section{Acknowledgements}

The author would like to thank Vladimir Sidorenko, Martin Bossert and Daniel Augot for discussions on Power decoding and this paper.
The author gratefully acknowledges the support of the Digiteo foundation, project \href{http://idealcodes.github.io/}{\color{navy}{IdealCodes}} while he was with Inria, and also, while the author was with Ulm University, the support of the German Research Council “Deutsche Forschungsgemeinschaft” (DFG) under grant BO 867/22-1.

\bibliographystyle{abbrv}
\bibliography{bibtex,add}

\end{document}               % End of document